\documentclass[prx,twocolumn,preprintnumbers,superscriptaddress,amsmath,amssymb,nofootinbib,nobalancelastpage]{revtex4-2}
\usepackage{graphicx, xcolor}
\graphicspath{{figures/}}
\usepackage{amsfonts, amsthm, microtype, mathrsfs, bbm}
\usepackage[colorlinks,allcolors=blue]{hyperref}
\usepackage{braket, mathtools, physics, enumerate}
\usepackage[capitalize]{cleveref}  

\newtheorem{theorem}{Theorem}
\newtheorem{lemma}{Lemma}
\DeclarePairedDelimiter\ceil{\lceil}{\rceil}

\newcommand{\1}{\mathbbm{1}}
\renewcommand{\O}{\mathcal O}
\newcommand{\Q}{\mathcal Q}

\newcommand{\dagg}{^\dagger}
\newcommand{\eps}{\epsilon}

\renewcommand{\ket}[1]{|#1\rangle}
\renewcommand{\bra}[1]{\langle#1|}
\DeclareMathOperator{\polylog}{polylog}
\DeclareMathOperator{\poly}{poly}
\newtheorem*{T1}{Theorem~\ref{th:1}}
\newtheorem*{L1}{Lemma~\ref{lm:1}'}

\definecolor{purple}{rgb}{.6,.1,.6}
\definecolor{darkgreen}{rgb}{.1,.6,.1}

\usepackage{tikz}
\usetikzlibrary{decorations.pathreplacing,calligraphy,decorations.markings}

\definecolor{tensor}{rgb}{0.5,0.8,0.5}
\definecolor{isometry}{rgb}{0.8,0.8,1}
\definecolor{unitary}{rgb}{0.8,0.5,.5}
\definecolor{gate}{rgb}{1.0,1.0,1.0}

\newcommand{\ATensor}[2]{
	\begin{scope}[shift={(#1)}]
		\draw (-1,0) -- (1,0);
		\draw (0,1) -- (0,0);
		\filldraw[fill=tensor] (-1/2,-1/2) -- (-1/2,1/2) -- (1/2,1/2) -- (1/2,-1/2) -- (-1/2,-1/2);
		\draw (0,0) node {\scriptsize #2};
	\end{scope}
}
\newcommand{\ADaggTensor}[2]{
	\begin{scope}[shift={(#1)}]
		\draw (-1,0) -- (1,0);
		\draw (0,-1) -- (0,0);
		\filldraw[fill=tensor,shift={(0,0)}] (-1/2,-1/2) -- (-1/2,1/2) -- (1/2,1/2) -- (1/2,-1/2) -- (-1/2,-1/2);
		\draw (0,0) node {\scriptsize #2};
	\end{scope}
}
\newcommand{\PTensor}[1]{
	\begin{scope}[shift={(#1)}]
		\draw (-1.5,0) -- (1.5,0);
		\draw (-.6, 0) -- (-.6, 1);https://www.overleaf.com/project/6464b177471e6194f64cb571
		\draw ( .6, 0) -- ( .6, 1);
		\filldraw[fill=tensor] (-1,-1/2) -- (-1,1/2) -- (1,1/2) -- (1,-1/2) -- (-1,-1/2);
		\draw (0,0) node {\scriptsize $P$};
	\end{scope}
}
\newcommand{\BTensor}[2]{
	\begin{scope}[shift={(#1)}]
		\draw (-1,0) -- (1,0);
		\foreach \x in {0,1,...,3}{
			\draw[shift={(-.3+0.2*\x,0)}] (0,1) -- (0,0);
		}
		\filldraw[fill=tensor] (-1/2,-1/2) -- (-1/2,1/2) -- (1/2,1/2) -- (1/2,-1/2) -- (-1/2,-1/2);
		\draw (0,0) node {\scriptsize #2};
	\end{scope}
}
\newcommand{\isometry}[2]{
	\begin{scope}[shift={(#1)}]
		\foreach \x in {0,1,...,3}{
			\draw[shift={(-.6+0.4*\x,0)}] (0,1) -- (0,0);
		}
		\foreach \x in {0,1,...,1}{
			\draw[shift={(-.6+1.2*\x,0)}] (0,-1) -- (0,0);
		}
		\filldraw[fill=isometry] (-1,-1/2) -- (-1,1/2) -- (1,1/2) -- (1,-1/2) -- (-1,-1/2);
		\draw (0,0) node {\scriptsize #2};
	\end{scope}
}
\newcommand{\unitary}[1]{
	\begin{scope}[shift={(#1)}]
		\def\xscale{2}
		\foreach \x in {0,1,...,3}{
			\draw[shift={(-.6*\xscale+0.4*\x*\xscale,0)}] (0,1) -- (0,-1);
		}
		\filldraw[fill=unitary] (-\xscale,-1/2) -- (-\xscale,1/2) -- (\xscale,1/2) -- (\xscale,-1/2) -- (-\xscale,-1/2);
		\draw (0,0) node {\scriptsize$U$};
		\draw (-.2*\xscale,-1.3) node {\scriptsize$0$};
		\draw (+.2*\xscale,-1.3) node {\scriptsize$0$};
	\end{scope}
}

\newcommand\subsetsim{\mathrel{%
  \ooalign{\raise0.2ex\hbox{$\subset$}\cr\hidewidth\raise-0.8ex\hbox{\scalebox{0.9}{$\sim$}}\hidewidth\cr}}}

\newcommand{\PEmptyTensor}[1]{
	\begin{scope}[shift={(#1)}]
		\draw (-1.5,0) -- (1.5,0);
		\draw (-.6, 0) -- (-.6, 1);
		\draw ( .6, 0) -- ( .6, 1);
		\filldraw[fill=tensor] (-1,-1/2) -- (-1,1/2) -- (1,1/2) -- (1,-1/2) -- (-1,-1/2);
	\end{scope}
}

\newcommand{\BDaggTensor}[2]{
	\begin{scope}[shift={(#1)}]
		\draw (-1,0) -- (1,0);
		\foreach \x in {0,1,...,3}{
			\draw[shift={(-.3+0.2*\x,-1)}] (0,1) -- (0,0);
		}
		\filldraw[fill=tensor] (-1/2,-1/2) -- (-1/2,1/2) -- (1/2,1/2) -- (1/2,-1/2) -- (-1/2,-1/2);
		\draw (0,0) node {\scriptsize #2};
	\end{scope}
}

\newcommand{\PDaggTensor}[2]{
	\begin{scope}[shift={(#1)}]
		\draw (-1.5,0) -- (1.5,0);
		\draw (-.6, 0) -- (-.6, -1);
		\draw ( .6, 0) -- ( .6, -1);
		\filldraw[fill=tensor] (-1,-1/2) -- (-1,1/2) -- (1,1/2) -- (1,-1/2) -- (-1,-1/2);
		\draw (0,0) node {\scriptsize $#2$};
	\end{scope}
}

\newcommand{\BTensorDense}[2]{
	\begin{scope}[shift={(#1)}]
		\draw (-1,0) -- (1,0);
		\foreach \x in {0,1,...,4}{
			\draw[shift={(-.35+0.175*\x,0)}] (0,1) -- (0,0);
		}
		\filldraw[fill=tensor] (-1/2,-1/2) -- (-1/2,1/2) -- (1/2,1/2) -- (1/2,-1/2) -- (-1/2,-1/2);
		\draw (0,0) node {\scriptsize #2};
	\end{scope}
}

\newcommand{\isometryDense}[2]{
	\begin{scope}[shift={(#1)}]
		\foreach \x in {0,1,...,4}{
			\draw[shift={(-.7+0.35*\x,0)}] (0,1) -- (0,0);
		}
		\foreach \x in {0,1,...,1}{
			\draw[shift={(-.6+1.2*\x,0)}] (0,-1) -- (0,0);
		}
		\filldraw[fill=isometry] (-1,-1/2) -- (-1,1/2) -- (1,1/2) -- (1,-1/2) -- (-1,-1/2);
		\draw (0,0) node {\scriptsize #2};
	\end{scope}
}

\begin{document}
\title{Preparation of matrix product states with log-depth quantum circuits}
\author{Daniel Malz}
\thanks{These authors (listed alphabetically) contributed equally to this work.}
\affiliation{Department of Mathematical Sciences, University of Copenhagen, Universitetsparken 5, 2100 Copenhagen, Denmark}

\author{Georgios Styliaris}
\thanks{These authors (listed alphabetically) contributed equally to this work.}
\affiliation{
Max-Planck-Institut f{\"{u}}r Quantenoptik, Hans-Kopfermann-Str. 1, 85748 Garching, Germany
}%
\affiliation{
Munich Center for Quantum Science and Technology (MCQST), Schellingstr. 4, 80799 M{\"{u}}nchen, Germany
}%

\author{Zhi-Yuan Wei}
\thanks{These authors (listed alphabetically) contributed equally to this work.}
\affiliation{
Max-Planck-Institut f{\"{u}}r Quantenoptik, Hans-Kopfermann-Str. 1, 85748 Garching, Germany
}
\affiliation{
Munich Center for Quantum Science and Technology (MCQST), Schellingstr. 4, 80799 M{\"{u}}nchen, Germany
}%

\author{J.~Ignacio Cirac}
\affiliation{
Max-Planck-Institut f{\"{u}}r Quantenoptik, Hans-Kopfermann-Str. 1, 85748 Garching, Germany
}%
\affiliation{
Munich Center for Quantum Science and Technology (MCQST), Schellingstr. 4, 80799 M{\"{u}}nchen, Germany
}%

\date{\today}

\begin{abstract}
    We consider the preparation of matrix product states (MPS) on quantum devices via quantum circuits of local gates.
    We first prove that faithfully preparing translation-invariant normal MPS of $N$ sites requires a circuit depth $T=\Omega(\log N)$.
    We then introduce an algorithm based on the renormalization-group transformation to prepare normal MPS with an error $\epsilon$ in depth $T=O(\log (N/\epsilon))$, which is optimal.
    We also show that measurement and feedback leads to an exponential speedup of the algorithm, to $T=O(\log\log (N/\epsilon))$.
    Measurements also allow one to prepare arbitrary translation-invariant MPS, including long-range non-normal ones, in the same depth.
    Finally, the algorithm naturally extends to inhomogeneous MPS.
\end{abstract}
\maketitle

One of the most important tasks in many-body physics and quantum information science is the preparation of useful or relevant states.
This has spurred a large effort to find ways to prepare states, for example, adiabatically~\cite{Albash2018}, dissipatively~\cite{Kraus2008,Verstraete2009}, or using quantum circuits.
A natural class of states to consider are matrix product states (MPS), because they efficiently approximate ground states of gapped local Hamiltonians~\cite{fannes1992finitely,Verstraete2006,Hastings2007}.
Moreover, many paradigmatic states can neatly be expressed as MPS, such as the cluster~\cite{Briegel2001}, GHZ~\cite{greenberger1989going}, $W$~\cite{dur2000three} and AKLT states~\cite{Affleck1987, Affleck1988}.

Several ways are known to prepare MPS.
Using unitary quantum circuits with strictly local gates,
all MPS can be prepared using a sequential quantum circuit
of depth $T\propto N$~\cite{Schoen2005}.
This is provably optimal for long-range correlated states such as the GHZ state~\cite{bravyi2006lieb}. However, for so-called normal MPS \cite{perez2007matrix}, which have short-range correlations, shorter depths are possible. Indeed, when allowing for a small error $\eps$, they can be obtained by acting on a product state with a constant-depth circuit of quasilocal gates---gates whose support grows (poly-)logarithmically with system size~\cite{Brandao2019,Piroli2021}.
However, such quasilocal gates have to be compiled into gates with strictly local support, and in the worst case such a compilation leads to circuits with a depth scaling exponentially in the support, and thus as $\poly(N)$.
However, since normal MPS all lie in the topologically trivial phase, one can construct adiabatic paths with a guaranteed gap~\cite{Schuch2011}, which means normal MPS can provably be prepared adiabatically in $T=O(\polylog (N/\eps))$~\cite{Ge2016} (also see~\cite{Bachmann2018}).

Despite of these results, it remains unclear if the scaling of the state-of-the-art algorithm~\cite{Ge2016} is optimal, or if there exist even faster algorithms to prepare normal MPS.
Proving optimality requires finding a tight lower bound on the depth, or, equivalently, its complexity, which is believed to be difficult in general~\cite{bravyi2006lieb,brandao2021models,jia2023hay}.

Here we first resolve the question of asymptotically optimal preparation of normal translation-invariant (TI) MPS. We prove that any circuit faithfully preparing them requires a depth $T=\Omega(\log N )$, i.e., it has to scale at least logarithmically with $N$. We then introduce an algorithm that saturates this bound and prepares all normal TI-MPS in a circuit depth
\begin{equation} \label{dep_normal}
    T=O(\log (N/\eps))
\end{equation}
using strictly local gates.
This is asymptotically faster than the previously fastest known algorithm (adiabatic preparation~\cite{Ge2016}) and also asymptotically optimal.
Moreover, the algorithm naturally extends to inhomogeneous MPS that are suitably short-range correlated.

\begin{figure}[t]
		\centering
    \begin{gather*}
        \begin{array}{c}
		\begin{tikzpicture}[scale=.4,thick,baseline={([yshift=-6ex]current bounding box.center)}]
			\foreach \x in {0,1,...,2}{
                \BTensorDense{1.5*\x,0,0}{}
			}
       \draw[dotted] (-1,  0) -- (-2, 0);
       \draw[dotted] (4,  0) -- (5, 0);
			\draw (-1.5, 1.5) node {\textbf{(a)}};
   \end{tikzpicture}
   \end{array}
   \approx
    \begin{array}{c}
			\begin{tikzpicture}[scale=.4,thick]
				\foreach \x in {0,1,...,1}{
                    \isometryDense{3.5*\x,1}{$V$}
					\draw (0.6+3.5*\x, 0) -- (3.5*\x + 2.9, 0);
					\filldraw[color=black, fill=white, thick](3.5*\x+1.75, 0) circle (0.3);
				}
                    \isometryDense{7,1}{$V$}
					\draw (0.6+7, 0) -- (7 + 2.4, 0);
					\filldraw[color=black, fill=white, thick](7+1.75, 0) circle (0.3);
					\draw (-1, 0) -- (-0.6, 0);
				\draw[dotted] (-1,  0) -- (-1.8, 0);
				\draw[dotted] (9.4,  0) -- (10.2, 0);
        \end{tikzpicture}
        \end{array}
\\
      	\begin{array}{c}
		\begin{tikzpicture}[scale=.4,thick,baseline={([yshift=-9ex]current bounding box.center)}]
		\foreach \x in {0,1,...,7}{
			\draw[shift={(-.7+0.35*\x,0)}] (0,1) -- (0,0);
		}
		\foreach \x in {0,1,...,1}{
			\draw[shift={(-.6+2.2*\x,0)}] (0,-1) -- (0,0);
		}
		\filldraw[fill=isometry] (-1,-1/2) -- (-1,1/2) -- (2,1/2) -- (2,-1/2) -- (-1,-1/2);
		\draw (0.5,0) node {\scriptsize $V$};
			\draw (-.8, 2) node {\textbf{(b)}};
		\end{tikzpicture}
        \end{array}
        =
        \begin{array}{c}
		\begin{tikzpicture}[scale=.4,thick,baseline={([yshift=-10ex]current bounding box.center)}]
  	\foreach \x in {0,1,...,1}{
					\begin{scope}[shift={(3*\x, 0)}]
						\draw (-1, 2);
						\draw (-0,0.5) -- (-0,0);
						\draw (+2/3,0.5) -- (+2/3,-2/3);
						\draw (+4/3,0.5) -- (+4/3,-2/3);
						\draw (+2,0.5) -- (+2,0);
                        \draw (2/3,-1.05) node {\scriptsize$0$};
                        \draw (4/3,-1.05) node {\scriptsize$0$};
      \filldraw[fill=unitary] (-1/3,-1/2) -- (-1/3,0) -- (2.0+1/3,0) -- (2.0+1/3,-1/2) -- (-1/3,-1/2);
					\end{scope}
				}
				\draw (0, -1.5) -- (0, -0.5);
				\draw[shift={(2, 0)}] (0, -1.5) -- (0, -0.5);
				\draw[shift={(3, 0)}] (0, -1.5) -- (0, -0.5);
				\draw[shift={(5, 0)}] (0, -1.5) -- (0, -0.5);
                \draw (0.0, -1.5) -- (0.0, -2.5);
				\draw (5, -1.5) -- (5, -2.5);
                \draw (2, -1.5) -- (2, -2.2);
				\draw (3, -1.5) -- (3, -2.2);
                \draw (2,-2.6) node {\scriptsize$0$};
                \draw (3,-2.6) node {\scriptsize$0$};
				\filldraw[fill=unitary] (-1/3,-2) -- (-1/3,-1.5) -- (5+1/3,-1.5) -- (5+1/3,-2) -- (-1/3,-2);
		\end{tikzpicture}
        \end{array}
\quad \;
\begin{array}{c}
		\begin{tikzpicture}[scale=.4,thick,baseline={([yshift=-9ex]current bounding box.center)}]
  					\draw (1.0, -0.75) -- (6.0, -.75);
					\draw (1.75, -0.75) -- (1.75, 0);
					\draw (5.25, -0.75) -- (5.25, 0);
					\draw (1.0+3.5*0, 0) -- (3.5*0 + 2.9, 0);
                    \draw[dotted] (0.0+3.5*0, 0) -- (1.0 + 3.5*0, 0);
					\filldraw[color=black, fill=white, thick](3.5*0+1.75, 0) circle (0.3);
					\filldraw[color=black, fill=black, thick](3.5*0+1.75, -.75) circle (0.1);
					\draw (0.6+3.5*1, 0) -- (3.5*1 + 2.5, 0);
                    \draw[dotted] (2.5+3.5*1, 0) -- (3.5*1 + 3.5, 0);
     \isometryDense{3.5*1,1}{$V$}
					\filldraw[color=black, fill=white, thick](3.5*1+1.75, 0) circle (0.3);
					\filldraw[color=black, fill=black, thick](3.5*1+1.75, -.75) circle (0.1);
    		          \draw[dotted] (0,  -.75) -- (1, -.75);
                    \draw[dotted] (6.0,  -.75) -- (7.0, -.75);
			\draw (0, 2) node {\textbf{(c)}};
        \filldraw[color=black, fill=white, thick](2.5, -0.3) rectangle (4.5, -1.2);
        \draw (3.5,-0.75) node {\scriptsize GHZ};
    \end{tikzpicture}
        \end{array}
    \end{gather*}
    \caption{Algorithm for MPS preparation. \textbf{(a)} After blocking, we approximate the state through its RG fixed point (nearest-neighbor entangled pairs for normal tensors [\cref{eq:normal_fp_local}]) combined with isometries $V$ that encode the local structure of the state. \textbf{(b)} We use RG to construct an efficient circuit for $V$, which can be further expressed with a low-depth circuit of local gates. \textbf{(c)} Our algorithm extends to non-normal (i.e., long-range correlated) MPS with GHZ-like fixed points [\cref{eq:fp_general_ti}] depicted here. Using quantum circuits assisted by measurements, both the fixed-point states and the isometry $V$ can be implemented efficiently.}\label{figure}\end{figure}
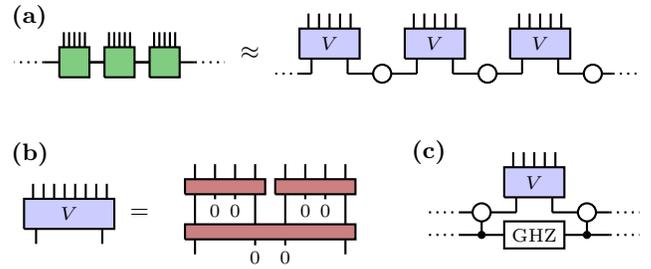

If one has additionally access to measurements and feedback, it is known that MPS can be prepared exactly in a depth $T=O(\log N)$ by expressing them in terms of 
the multiscale entanglement renormalization ansatz (MERA)~\cite{Lu2022}.
Including measurements also yields a speedup for our algorithm, and allows us to extend it to non-normal MPS, such that all TI-MPS can provably be prepared in depth
\begin{equation} \label{dep_meas}
    T=O(\log\log (N/\eps)).
\end{equation}
This is exponentially faster than the best known measurement assisted protocol~\cite{Lu2022}.
It also shows that our lower bound can be violated with access to measurements.
As a byproduct, our work also proves that the finite-range MERA~\cite{Evenbly2011} can approximate normal TI-MPS in $O(\log \log (N/\eps))$ layers.

Our algorithm fundamentally builds on the renormalization-group (RG) transformation.
The RG procedure consists of blocking neighboring sites and discarding short-range correlations. After consecutive RG transformations, the state asymptotically converges to its fixed-point state~\cite{Verstraete2005}, which has only nearest-neighbor correlations for normal TI-MPS~\cite{Verstraete2005,cirac2017matrix}.
This happens rapidly since it suffices to block only $O(\log (N/\eps))$ sites to approximate well the fixed-point state~\cite{Piroli2021}.
Our algorithm first prepares this fixed point, and subsequently reintroduces the short-range correlations by applying an isometry of support $\log (N/\eps)$ [cf.~\cref{figure}(a)].
Our key contribution is that we can prove through an explicit construction (inspired by earlier works~\cite{Lu2022,Schoen2005,Verstraete2005}) that this isometry can be implemented with a strictly local circuit of depth $T=O(\log (N/\eps))$ [cf.~\cref{figure}(b)]. When assisted by measurements [cf.~\cref{figure}(c)], the depth of the isometry can be further reduced, while the GHZ-like fixed point of long-range correlated MPS can be prepared in constant depth~\cite{briegel2001persistent,Piroli2021}. Together, this lead to the circuit depth $T=O(\log\log (N/\eps))$ to prepare almost arbitrary (including all TI) MPS.

\emph{Preliminaries.---}%
For simplicity, we first consider (normalized) TI-MPS,
\begin{align}
	\ket{\phi_N} \propto \sum_{i_1, \ldots, i_N=1}^d\Tr\left( A^{i_1}\cdots A^{i_N} \right)\ket{i_1\cdots i_N},
	\label{eq:TI-MPS}
\end{align}
and later extend to the inhomogeneous case. Above $A^i$ are $D\times D$ matrices ($D$ is the bond dimension) with $i=1,\dots,d$ (physical dimension).
We will extensively use graphical notation and identify
${(A^i)_{jk} =
\begin{array}{c}
    \begin{tikzpicture}[scale=.4, baseline={([yshift=-5.5ex]current bounding box.center)}, thick]
    	\ATensor{0,0}{$A$}
    	\draw (0,1.35) node {\scriptsize $i$};
    	\draw (1.35,0) node {\scriptsize $k$};
    	\draw (-1.35,0) node {\scriptsize $j$};
    \end{tikzpicture}
\end{array}}.$

To each tensor $A$ we associate its \textit{transfer matrix}
\begin{equation}
	E_A = \sum_{i=1}^d (A^{i})^* \otimes A^i
	=
  	\begin{array}{c}
		\begin{tikzpicture}[scale=.5,thick,baseline={([yshift=1ex]current bounding box.center)}]
			\ATensor{0,0}{$A$}
			\ADaggTensor{0,1.7}{$A^*$}
		\end{tikzpicture}
	\end{array} .
	\label{eq:transfer_matrix}
\end{equation}
A tensor is called \textit{normal}, if (i) it is irreducible ($A^i$ have no nontrivial common invariant subspace), and (ii)  $E_A$ has a unique largest eigenvalue $\lambda_1 = 1$ and no other of the same magnitude~\cite{fannes1992finitely,perez2007matrix}. 
Its correlation length is defined via the subleading eigenvalue $\xi = -1/\ln(|\lambda_2|)$.
After a gauge transformation~\cite{perez2007matrix}, $E_A$ of a normal tensor can be brought into the form
\begin{equation}
	E_A = \ket{\rho}\bra{\1} + R=
  	\begin{array}{c}
		\begin{tikzpicture}[scale=.4,thick,baseline={([yshift=1ex]current bounding box.center)}]
			\draw[shift={(-0.3, 0)}] (-1, 1) -- (-0.5, 1) -- (-0.5, -1) -- (-1, -1);
			\draw[shift={(-0.2, 0)}] (+1, 1) -- (+0.5, 1) -- (+0.5, -1) -- (+1, -1);
			\filldraw[color=black, fill=white, thick](-0.8, 0) circle (0.6);
			\draw (-0.8, 0) node {\scriptsize $\rho$};
		\end{tikzpicture}
    \end{array}
    + R
	,
	\label{eq:Ek_decomp}
\end{equation}
where the leading right eigenvector $\rho > 0 $ (Hermitian and positive definite)~\cite{evans1977spectral,perez2007matrix}, $\langle \1|\rho\rangle = \Tr (\rho) =  1$, and $R$ has spectral radius less than 1.

Blocking $q$ sites together yields a new tensor $B$
\begin{equation}
  	\begin{array}{c}
		\begin{tikzpicture}[scale=.45,thick,baseline={([yshift=-3ex]current bounding box.center)}]
			\BTensor{0,0}{$B$}
		\end{tikzpicture}
	\end{array}
	= 
  	\begin{array}{c}
	\begin{tikzpicture}[scale=.45, baseline={([yshift=5.5ex]current bounding box.center)}, thick]
		\draw[shift={(0,0)},dotted] (0,0) -- (4,0);
		\ATensor{0,0}{$A$}
		\ATensor{4,0}{$A$}
		\draw [decorate,
    	decoration = {calligraphic brace,mirror}] (0,-0.8) --  (4,-0.8);
		\draw (2,-1.5) node {\scriptsize $q$};
	\end{tikzpicture}
	\end{array}
	\label{eq:B}
\end{equation}
with physical dimension $d^q$, the same bond dimension $D$, and transfer matrix $E_B = E_A^q$.
$E_B$ approaches  its fixed point in the limit $q \to \infty$~\cite{Verstraete2005,cirac2017matrix} which, for normal tensors, is $E_{\infty} = \ket{\rho}\bra{\1}$. 

Our goal is to devise an algorithm that approximates the target $N$-site MPS $\ket{\phi_N}$ by $| \widetilde \phi_N \rangle$
with error $\eps = \eps(\phi_N, \widetilde\phi_N)$, where
\begin{equation}
    \eps(\phi,\psi) = 1-|\bra\phi\psi\rangle|
\end{equation}
and $| \widetilde \phi_N \rangle$ is prepared using a local quantum circuit.
Our first result is that it is impossible to approximate well normal TI-MPS in depth $o(\log N)$. Subsequently, we provide an explicit algorithm with the asymptotically optimal depth $O(\log (N/\eps))$.

\emph{Lower bound.---}%
Given (i) $\{\ket{\phi_N}\}$,
    a sequence of normalized TI-MPS on $N$ sites, generated by a normal tensor $A$, with finite correlation length $\xi>0$ [cf.~\cref{eq:TI-MPS}],
and (ii) $\{\ket{\psi_N}\}$, 
    a sequence obtained from depth-$T$ local quantum circuits applied to product states,
we are interested in determining how fast $T$ has to grow in order to approximate the MPS well, as measured by the error $\eps=\eps(\phi_N,\psi_N)$.
We prove here that no quantum circuit with depth $T=o(\log N)$ can faithfully approximate this class.
\begin{theorem} 
\label{th:1}
    If $T=o(\log N)$ there is some $N_0$ such that for all $N>N_0$ we have $\eps >1/2$.
\end{theorem}

The proof can be found in~\cite{sm}.
To establish this result, we use the fact that $\ket{\psi_N}$ have a strictly finite light cone, whereas in a normal TI-MPS $\ket{\phi_N}$ correlation functions decay only exponentially. 
This leads to a mismatch in the expectation value of correlators outside the light cone, which gives a lower bound on the error between the two states.
We additionally use the fact that sufficiently distant parts of the system are statistically independent, such that the error accumulates with increasing system size $N$, unless the circuit depth grows sufficiently quickly.

\emph{The algorithm.---}%
We now present the key steps for our algorithm. We will (i) approximate $\ket{\phi_N}$ by $\ket{\widetilde \phi_N}$, then (ii) show that $\ket{\widetilde \phi_N}$ can be efficiently prepared, and (iii) prove that the approximation error decays sufficiently fast with $N$. We begin with the case of normal TI-MPS and return to the general case later.

\emph{Approximation through the fixed-point state.---}To make the approximation, we follow the steps of the RG transformation~\cite{Verstraete2005}.
After blocking $q$ sites, we perform a polar decomposition on the blocked tensor $B$, interpreting it as a map from the $D^2$-dimensional virtual space to the $d^q$-dimensional physical space\footnote{To do a polar decomposition we need $q$ large enough that $d^q\geq D^2$.
Later, we will have $q$ scaling with the system size, so here we assume that $B$ is injective~\cite{perez2007matrix}.
This is always true for normal tensors after blocking finite (independent of $N$) sites~\cite{Sanz2010}.}.
This way we can write $B = V P$ where $V$ is an isometry with $V^\dagger V = \1_{D^2}$ and $P > 0$ is positive definite.
Thus
\begin{align} \label{eq:B_TM}
    E_B = 
    \begin{array}{c}
		\begin{tikzpicture}[scale=.39,thick,baseline={([yshift=1ex]current bounding box.center)}]
			\BTensor{0,0}{$B$}
			\BDaggTensor{0,1.7}{$B^*$}
		\end{tikzpicture}
	\end{array}
  =
      \begin{array}{c}
		\begin{tikzpicture}[scale=.39,thick,baseline={([yshift=1ex]current bounding box.center)}]
			\PTensor{0,0}
			\PDaggTensor{0,1.7}{P^*}
		\end{tikzpicture}
	\end{array}
 \xrightarrow{q \to \infty}
       \begin{array}{c}
		\begin{tikzpicture}[scale=.4,thick,baseline={([yshift=1ex]current bounding box.center)}]
			\PEmptyTensor{0,0}
			\PDaggTensor{0,1.7}{P_{\infty}^*}
		\draw (0,0) node {\scriptsize $P_{\infty}$};
		\end{tikzpicture}
	\end{array}
 =   
 		\begin{tikzpicture}[scale=.39,thick,baseline={([yshift=-0.52ex]current bounding box.center)}]
			\draw[shift={(-0.3, 0)}] (-1, 1) -- (-0.5, 1) -- (-0.5, -1) -- (-1, -1);
			\draw[shift={(-0.2, 0)}] (+1, 1) -- (+0.5, 1) -- (+0.5, -1) -- (+1, -1);
			\filldraw[color=black, fill=white, thick](-0.8, 0) circle (0.6);
			\draw (-0.8, 0) node {\scriptsize $\rho$};
		\end{tikzpicture}
  .
\end{align}

The approximation consists of replacing $P$ by its fixed-point version $P_\infty$ in the tensor $B$, while keeping the isometry $V$ intact.
Graphically,
\begin{equation}
  	\begin{array}{c}
		\begin{tikzpicture}[scale=.5,thick]
			\BTensor{0,0}{$B$}
		\end{tikzpicture}
	\end{array}
	=
  	\begin{array}{c}
		\begin{tikzpicture}[scale=.4,thick]
			\PTensor{0,0}
			\isometry{0,1.5}{$V$}
		\end{tikzpicture}
	\end{array}
	\approx
  	\begin{array}{c}
		\begin{tikzpicture}[scale=.4,thick]
					\draw (-2.4, 0) -- (-0.6, 0);
					\filldraw[color=black, fill=white, thick](-1.5, 0) circle (0.6);
                    \draw (0.6, 0) -- (1.8, 0);
			\isometry{0 ,1}{$V$}
			\draw (-1.6,0) node {\scriptsize $\sqrt{\rho}$};
		\end{tikzpicture}
    \end{array}
	=
  	\begin{array}{c}
		\begin{tikzpicture}[scale=.5,thick]
			\BTensor{0,0}{$\widetilde B$}
		\end{tikzpicture}
	\end{array} .
	\label{eq:key_approximation}
\end{equation}
Later we will assign meaning to the approximation sign in \cref{eq:key_approximation} by bounding the global error between the MPS $\ket{\phi_N}$ and $| \widetilde \phi_N \rangle$ resulting from the two tensors, $B$ and $\widetilde B$. To obtain a vanishing error in the thermodynamic limit we will need $q \propto \log N$, which we assume for now and justify subsequently.

\emph{Preparing the approximate state.---}The approximate state $|\widetilde\phi_N\rangle$ can be prepared by acting on the fixed-point state with a product of unitaries of support $q$ (for simplicity $D = d$ in the illustration)
\begin{equation} \label{eq:phi_tilde}
		\begin{aligned}
    	|\widetilde \phi_N \rangle = \left(\bigotimes_{i=1}^{N/q} U_i\right)\bigotimes_{i = 1} ^{N/q} \left( \ket{\omega}_{R_i L_{i+1}} \ket{0 \dots 0}_{C_i} \right)\\
  		=
  		\begin{array}{c}
			\begin{tikzpicture}[scale=.5,thick]
				\foreach \x in {0,1,...,1}{
					\unitary{5*\x, 1}
					\draw (1.2+5*\x, 0) -- (5*\x + 3.8, 0);
					\filldraw[color=black, fill=white, thick](5*\x+2.5, 0) circle (0.3);
				}
				\draw (-1.2, 0) -- (-2.1, 0);
				\draw[dotted] (9,  0) -- (10, 0);
				\draw[dotted] (-2.2,  0) -- (-3.3, 0);
			\draw (-1.3,-1.2) node {\scriptsize $L_i$};
   			\draw (0.05,-1.2) node {\scriptsize $C_i$};
			\draw (1.5,-1.2) node {\scriptsize $R_i$};
   			\draw (3.7,-1.25) node {\scriptsize $L_{i+1}$};
   			\draw (5.15,-1.25) node {\scriptsize $C_{i+1}$};
   			\draw (6.75,-1.25) node {\scriptsize $R_{i+1}$};
			\end{tikzpicture}
    	\end{array}.
	\end{aligned}
\end{equation}
The unitary is constructed such that it implements the required isometry when acting on a product state\footnote{From dimension counting $D^2 d^\ell \ge d^q$ thus $\ell \sim q$.}  $\ket{0}^{\otimes \ell}$ over the ``central'' region ($\ell = 2$ in the illustration)
\begin{equation}
  	\begin{array}{c}
		\begin{tikzpicture}[scale=.4,thick,baseline={([yshift=3.8ex]current bounding box.center)}]
			\unitary{0,0}
		\end{tikzpicture}
    \end{array}
    =
  	\begin{array}{c}
		\begin{tikzpicture}[scale=.4,thick,baseline={([yshift=0ex]current bounding box.center)}]
			\isometry{0,0}{$V$}
		\end{tikzpicture}
    \end{array}.
	\label{eq:unitary}
\end{equation}
Note that for normal TI-MPS the fixed-point state $\ket{\Omega} = \otimes_{i=1}^{N/q} \ket{\omega}_{R_i L_{i+1}}$ is a tensor product of entangled pairs,
\begin{align} \label{eq:normal_fp_local}
   |\omega\rangle_{R_i L_{i+1}} =
    \begin{array}{c}
			\begin{tikzpicture}[scale=.45,thick,baseline={([yshift=4ex]current bounding box.center)}]
					\draw (-1, .5) -- (-1, 0) -- (1,0) -- (1,0.5);
					\filldraw[color=black, fill=white, thick](0, 0) circle (0.3);
			\draw (-.8,-0.7) node {\scriptsize $R_i$};
   			\draw (1.25,-0.75) node {\scriptsize $L_{i+1}$};
        \end{tikzpicture}
        \end{array} 
         \!\!\!\!\!  =
( \1 \otimes \sqrt{\rho}) \sum_{i=1}^D \ket{ii}_{R_i L_{i+1}}
\end{align}
each with support over the ``right'' and ``left'' Hilbert spaces of neighboring sites ($\dim R_i = \dim  L_{i+1} = D$). It can thus be prepared from a product state with a constant-depth circuit.

So far, it is not obvious that the resulting circuit can be expressed efficiently in terms of strictly local gates, because the unitaries in \cref{eq:unitary} are only \emph{quasilocal}, i.e., having support $q \propto \log N$.
While a naive bound on the circuit depth would be $\mathrm{poly}(N)$, here we use the fact that $U$ comes from an MPS to show that in reality it can be implemented in $T=O(q)$. We do this by providing two explicit and exact decompositions of $U$ in terms of gates with constant support, the ``sequential-RG'' and the ``tree-RG''.

\emph{The sequential-RG circuit.---}We can express the unitary in \cref{eq:phi_tilde} in terms of the original MPS by applying the inverse of $P$ to its virtual legs\footnote{The subsequent derivation remains valid also for non-injective tensors $B$. In that case $P^{-1}$ is understood as pseudo-inverse.},
\begin{align}
  	\begin{array}{c}
		\begin{tikzpicture}[scale=.4,thick,baseline={([yshift=4ex]current bounding box.center)}]
			\unitary{0,0}
		\end{tikzpicture}
    \end{array}
    =
  	\begin{array}{c}
		\begin{tikzpicture}[scale=.4,thick]
			\foreach \x in {0,1,...,3}{
				\ATensor{(1.5*\x,0)}{$A$}
			}
			\draw (-1, 0) -- (-1, -.8) -- (3.5, -.8) -- (3.5, -2.5);
			\draw (5.5, 0) -- (5.5, -.8) -- (4.5, -.8) -- (4.5, -2.5);
			\filldraw[fill=tensor, thick] (3, -1.1) -- (5, -1.1) -- (5, -2.1) -- (3, -2.1) -- (3, -1.1) ;
			\draw (4,-1.6) node {\scriptsize $P^{-1}$};
   			\draw (3.50, -2.5) -- (0, -2.5) -- (0, -2.8);
   			\draw (4.5, -2.1) -- (4.5, -2.8);
		\end{tikzpicture}
    \end{array}
    =
  	\begin{array}{c}
		\begin{tikzpicture}[scale=.4,thick,fill=tensor,baseline={([yshift=2ex]current bounding box.center)}]
			\draw (0,.1) -- (4.5,.1);
			\draw (0,-.1) -- (4.5,-.1);
			\foreach \x in {0,1,...,3}{
				\begin{scope}[shift={(1.5*\x,0)}]
					\draw (0,0) -- (0,1);
					\filldraw (-.5, -.5) -- (-.5,.5) -- (.5,.5) -- (.5,-.5) -- (-.5,-.5);
					\draw (0,0) node {\scriptsize $A'$};
				\end{scope}
			}
			\draw (4.4, 0) -- (4.4, -1);
			\draw (4.6, 0) -- (4.6, -1);
			\filldraw (4, -.5) -- (4,.5) -- (5,.5) -- (5,-.5) -- (4,-.5);
			\draw (4.5,0) node {\scriptsize $C$};
			\draw (4.4, -1) -- (0, -1) -- (0, -1.3);
			\draw (4.6, -0.5) -- (4.6, -1.3);
			\draw (-.5, 0.15 ) -- (-0.85, 0.15) -- (-0.85, -0.15) -- (-0.5, -0.15);
		\end{tikzpicture}
    \end{array},
	\label{eq:sequential_unitary}
\end{align}
where in the last step we set $A'^i = A^i \otimes \1_D$ and contracted $P^{-1}$ with the rightmost $A'$ to obtain $C$.
As in sequential preparation of MPS~\cite{Schoen2005} and in the left-canonical form~\cite{Schollwoeck2011}, we can now iteratively apply singular value decompositions, starting from the tensor on the left and moving right, but stopping before the last tensor\footnote{The method can easily be generalized to absorb $P^{-1}$ into any of the tensors.
With $C$ in the bulk, the sequential circuit is obtained by repeated SVD starting both left and right and stopping at $C$, as in the mixed canonical form~\cite{Schollwoeck2011}.}.
This defines a new set of tensors that describe the same isometry $V$ but now each tensor is a local isometry (arrows indicate isometry direction, $q=4$ in illustration)
\begin{align}
    V = V_{q} \dots V_1 =
  	\begin{array}{c}
		\begin{tikzpicture}[scale=.4,thick,fill=isometry,decoration={
		  		  	markings, mark=at position 0.75 with {\arrow{>}}}]
			\foreach \x in {1,...,2}{
				\draw[postaction={decorate}] (0,1) -- (0,.5);
				\draw[postaction={decorate}] (.5,0) -- (1,0);
				\filldraw[fill=isometry] (-.5, -.5) -- (-.5,.5) -- (.5,.5) -- (.5,-.5) -- (-.5,-.5);
				\draw (0,0) node {\scriptsize $V_4$};
				\begin{scope}[shift={(1.5*\x,0)}]
					\draw[postaction={decorate}] (0,1) -- (0,.5);
					\draw[postaction={decorate}] (.5,.1) -- (1,.1);
					\draw[postaction={decorate}] (.5,-.1) -- (1,-.1);
					\filldraw (-.5, -.5) -- (-.5,.5) -- (.5,.5) -- (.5,-.5) -- (-.5,-.5);
				\end{scope}
			\draw (1.5,0) node {\scriptsize $V_{3}$};
   			\draw (3,0) node {\scriptsize $V_{2}$};
			}
			\draw (4.5,1) -- (4.5,.5);
			\draw (4.4, 0) -- (4.4, -1);
			\draw (4.6, 0) -- (4.6, -1);
			\filldraw[fill=tensor] (4, -.5) -- (4,.5) -- (5,.5) -- (5,-.5) -- (4,-.5);
			\draw (4.5,0) node {\scriptsize $\widetilde C$};
		\end{tikzpicture}
    \end{array}
    , \quad V_i : \mathbb C^{D'_{i}} \to \mathbb C^{d D'_{i+1}} .
	\label{eq:sequential_svds}
\end{align}
with every $V_i$ an isometry $V_i^\dagger V_i = \1_{D'_i}$ satisfying $D'_i \le D^2$ ($D'_{q+1} = 1$).
Importantly, $\widetilde C = V_1$ is automatically also an isometry, as
\begin{equation}
    V\dagg V = 
  	\begin{array}{c}
		\begin{tikzpicture}[scale=.5,thick,fill=isometry]
			\draw (0,1) -- (0,.5);
			\draw (.5,0) -- (1,0);
			\filldraw (-.5, -.5) -- (-.5,.5) -- (.5,.5) -- (.5,-.5) -- (-.5,-.5);
			\draw (0,0) node {\scriptsize $V_4$};
			\foreach \x in {1,...,2}{
				\begin{scope}[shift={(1.5*\x,0)}]
					\draw (0,1) -- (0,.5);
					\draw (.5,.1) -- (1,.1);
					\draw (.5,-.1) -- (1,-.1);
					\filldraw (-.5, -.5) -- (-.5,.5) -- (.5,.5) -- (.5,-.5) -- (-.5,-.5);
				\end{scope}
			\draw (1.5,0) node {\scriptsize $V_{3}$};
			\draw (3,0) node {\scriptsize $V_{2}$};
			}
			\draw (4.5,1) -- (4.5,.5);
			\draw (4.4, 0) -- (4.4, -1);
			\draw (4.6, 0) -- (4.6, -1);
			\filldraw[fill=tensor] (4, -.5) -- (4,.5) -- (5,.5) -- (5,-.5) -- (4,-.5);
			\draw (4.5,0) node {\scriptsize $\widetilde C$};
			\begin{scope}[shift={(0, 1.5)}]
				\draw (.5,0) -- (1,0);
				\filldraw (-.5, -.5) -- (-.5,.5) -- (.5,.5) -- (.5,-.5) -- (-.5,-.5);
				\draw (0,0) node {\scriptsize $V_4^*$};
				\foreach \x in {1,...,2}{
					\begin{scope}[shift={(1.5*\x,0)}]
						\draw (.5,.1) -- (1,.1);
						\draw (.5,-.1) -- (1,-.1);
						\filldraw (-.5, -.5) -- (-.5,.5) -- (.5,.5) -- (.5,-.5) -- (-.5,-.5);
					\end{scope}
				}
				\draw (4.4, 0) -- (4.4, 1);
				\draw (4.6, 0) -- (4.6, 1);
				\filldraw[fill=tensor] (4, -.5) -- (4,.5) -- (5,.5) -- (5,-.5) -- (4,-.5);
				\draw (4.5,0) node {\scriptsize ${\widetilde C}^*$};
				\draw (1.5,0) node {\scriptsize $V_{3}^*$};
			    \draw (3,0) node {\scriptsize $V_{2}^*$};
			\end{scope}
		\end{tikzpicture}
    \end{array}
    =
  	\begin{array}{c}
		\begin{tikzpicture}[scale=.5,thick,fill=isometry]
			\draw (0, 0) -- (0, 1);
			\draw (0, -.1) -- (-.9, -.1) -- (-.9,1.6) -- (0, 1.6);
			\draw (0, .1) -- (-.7, .1) -- (-.7,1.4) -- (0, 1.4);
			\draw (.1, 0) -- (.1, -1);
			\draw (-.1, 0) -- (-.1, -1);
			\draw (.1, 1.5) -- (.1, 2.5);
			\draw (-.1, 1.5) -- (-.1, 2.5);
			\filldraw[fill=tensor] (-.5, -.5) -- (-.5,.5) -- (.5,.5) -- (.5,-.5) -- (-.5,-.5);
			\filldraw[fill=tensor,shift={(0,1.5)}] (-.5, -.5) -- (-.5,.5) -- (.5,.5) -- (.5,-.5) -- (-.5,-.5);
			\draw (0,0) node {\scriptsize $\widetilde C$};
			\draw (0,1.5) node {\scriptsize $\widetilde C^*$};
		\end{tikzpicture}
    \end{array}
    =
      	\begin{array}{c}
		\begin{tikzpicture}[scale=.4,thick,fill=isometry]
			\draw (-.2,-1.5) -- (-.2,1.5);
			\draw (.2,-1.5) -- (.2, 1.5);
		\end{tikzpicture}
    \end{array} .
	\label{eq:C_is_an_isometry}
\end{equation}
Since this sequential circuit comprises $q$ sites, its depth is $O(q)$. This scaling is unchanged if we additionally take into account that the inputs of the unitary in \cref{eq:sequential_unitary} are separated by $O(q)$ sites, which requires one to implement SWAP gates.

\emph{The tree-RG circuit.---}Blocking two neighboring sites followed by a polar decomposition is the basis for the real-space RG transformation and halves the correlation length~\cite{Verstraete2005,Levin2007,Vidal2007}.
Instead of directly blocking $q$ sites, we can repeatedly apply this transformation $k \sim \log_2 q$ times to the same effect (illustration below, and in \cref{figure}(b)).
This generates a treelike circuit with $k$ layers, in which each layer but the lowest consists of isometries from dimension $D^2$ to $D^4$ (below $A_k$ is obtained from $A$ by blocking $k$ sites, and $q=8$)
\begin{align}
\begin{array}{c}
    \begin{tikzpicture}[scale=.4,thick,baseline={([yshift=2ex]current bounding box.center)}]
        \foreach \x in {0,1,...,1}{
            \begin{scope}[shift={(3*\x, 0)}]
                \draw (-1,0) -- (2,0);
                \draw (-.1,1) -- (-.1,0);
                \draw (+.1,1) -- (+.1,0);
                \draw (+0.9,1) -- (+0.9,0);
                \draw (+1.1,1) -- (+1.1,0);
                \filldraw[fill=tensor] (-1/2,-1/2) -- (-1/2,1/2) -- (3/2,1/2) -- (3/2,-1/2) -- (-1/2,-1/2);
                \draw (0.5, 0) node {\scriptsize $A_4$};
            \end{scope}
        }
    \end{tikzpicture}
\end{array} 
=
\begin{array}{c}
    \begin{tikzpicture}[scale=.4,thick]
        \foreach \x in {0,1,...,1}{
            \begin{scope}[shift={(3*\x, 0)}]
                \PEmptyTensor{0.5,-1.5}
                \draw (0.25+0.25\x, -1.5) node {\scriptsize $P^{(1)}$};
                \draw (-1, 2);
                \draw (-.1,1) -- (-.1,0);
                \draw (+.1,1) -- (+.1,0);
                \draw (+0.9,1) -- (+0.9,0);
                \draw (+1.1,1) -- (+1.1,0);
                \filldraw[fill=isometry] (-1/2,-1/2) -- (-1/2,1/2) -- (3/2,1/2) -- (3/2,-1/2) -- (-1/2,-1/2);
                \draw (0.5, 0) node {\scriptsize $V^{(1)}$};
            \end{scope}
        }
    \end{tikzpicture}
\end{array}
&=
\begin{array}{c}
    \begin{tikzpicture}[scale=.3,thick]
        \foreach \x in {0,1,...,1}{
            \begin{scope}[shift={(3*\x, 0)}]
                \draw (-1, 2);
                \draw (-.1,1) -- (-.1,0);
                \draw (+.1,1) -- (+.1,0);
                \draw (+0.9,1) -- (+0.9,0);
                \draw (+1.1,1) -- (+1.1,0);
                \filldraw[fill=isometry] (-1/2,-1/2) -- (-1/2,1/2) -- (3/2,1/2) -- (3/2,-1/2) -- (-1/2,-1/2);
                \draw (0.5, 0) node {\scriptsize $V^{(1)}$};
            \end{scope}
        }
        \draw (-1, -3) -- (3+2, -3);
        \draw (0, -1.5) -- (0, -0.5);
        \draw[shift={(1, 0)}] (0, -1.5) -- (0, -0.5);
        \draw[shift={(3, 0)}] (0, -1.5) -- (0, -0.5);
        \draw[shift={(4, 0)}] (0, -1.5) -- (0, -0.5);
        \draw (0.0, -1.5) -- (0.0, -2.5);
        \draw (4, -1.5) -- (4, -2.5);
        \filldraw[fill=isometry] (-1/2,-2) -- (-1/2,-1) -- (3+3/2,-1) -- (3+3/2,-2) -- (-1/2,-2);
        \filldraw[fill=tensor, shift={(0, -1.5)}] (-1/2,-2) -- (-1/2,-1) -- (3+3/2,-1) -- (3+3/2,-2) -- (-1/2,-2);
        \draw (2, -1.5) node {\scriptsize $V^{(2)}$};
        \draw (2, -3) node {\scriptsize $P^{(2)}$};
    \end{tikzpicture}
\end{array}.
\label{eq:RG_circuit}
\end{align}
In \cref{eq:RG_circuit}, the lowest layer is again the part that is replaced by the fixed-point state in our algorithm, i.e., a product of $\ket{\omega}$ [cf.~\cref{eq:normal_fp_local}].
In this scheme, the lowest isometry $V^{(k)}$ acts across a distance $q$.
Though not strictly local, this can be done in a depth $O(q)$ utilizing SWAP gates. 
Subsequent isometries act over distances $q/2$, $q/4$ and so forth, leading to an overall circuit depth $T=O(q)$.

\emph{Approximation error.---}So far, we constructed efficient circuits for preparing $|\widetilde\phi_N\rangle$, having assumed that we block $q$ sites. The scaling of $q$ is a consequence of the following Lemma, which is adapted from Ref.~\cite{Piroli2021}.
\begin{lemma}
    \label{lm:1}
    Given a sequence of TI-MPS generated from a normal tensor, and for all $\gamma < 1/2$,
    \begin{equation} \label{eq:fid_err}
        \epsilon(\widetilde \phi_N, \phi_N) = O\left( \frac{N}{q} e^{-\gamma q/\xi} \right).
    \end{equation}
\end{lemma}
The proof can be found in~\cite{sm}.
Using \cref{lm:1}, it follows that $q = O (\log (N / \epsilon))$. In particular, blocking $q = \ceil*{ 2 \xi (1 + \eta) \ln N} \propto \log N $ sites gives $\eps = O(N^{- \eta})$ for any $\eta > 0 $.

We also numerically illustrate the exponential decay of \cref{eq:fid_err} in~\cite{sm} for preparing the 1D AKLT state~\cite{Affleck1987, Affleck1988} and an MPS family with tunable correlation length~\cite{Wolf2006}, which demonstrates that the circuit is also efficient in practice.

\emph{Inhomogeneous short-range correlated MPS.---}Our results can be straightforwardly extended to MPS that have a finite correlation length, but are not TI.
The setting here is that we are given a sequence of MPS $\{\ket{\phi_N}\}$ with bond dimension at most $D$. 
We define such a sequence to have finite correlation length if, after blocking $q = O (\log N)$ times, the resulting states can be approximated up to quasi-local isometries by a state consisting of nearest-neighbor entangled pairs $| \Omega \rangle = \bigotimes_{i=1}^{N/q}  \ket{\omega^i}_{R_i L_{i+1}} $,
with an error
$\eps(\Omega,\phi_\mathrm{pos}) \to 0$ as  $N \to \infty$.
Here
\begin{align}
\label{eq:phi_pos}
    | \phi_{{{\rm pos}}} \rangle = \begin{array}{c}
    \begin{tikzpicture}[scale=.4,thick]
        \foreach \x in {0,1,...,2}{
            \begin{scope}[shift={(3*\x, 0)}]
                \PEmptyTensor{0.5,-1.5}
            \end{scope}
        }
        \draw (0.50, -1.5) node {\scriptsize $P_{i-1}$};
        \draw (3.50, -1.5) node {\scriptsize $P_{i}$};
        \draw (6.50, -1.5) node {\scriptsize $P_{i+1}$};
        \draw[dotted] (-2.1,  -1.5) -- (-1.1, -1.5);
        \draw[dotted] (8.1,  -1.5) -- (9.1, -1.5);
    \end{tikzpicture}
\end{array}
\end{align}
arises after blocking $q$ sites and keeping the positive part of the decomposition of $\ket{\phi_N}$.
If the finite correlation assumption is satisfied, then the preparation scheme consists of preparing $|  \Omega \rangle$ and implementing the isometry, decomposed with either of the two methods. The resulting total depth is again $O(\log (N/\eps))$ with error $\eps(\Omega,\phi_\mathrm{pos})$, as in the TI case.

In the Supplemental Material~\cite{sm} we numerically show that this protocol can prepare inhomogeneous random MPS~\cite{garnerone2010typicality,haferkamp2021emergent,lancien2022correlation,Haag2023typical} efficiently.
For that, we use a simple extension of the Evenbly-Vidal algorithm~\cite{evenbly2009algorithms,lin2021real}, which efficiently variationally finds $|\Omega \rangle$.
This illustrates that our finite correlation length assumption, as defined earlier, is satisfied in a practical setting.

\emph{Preparations using measurements.---}Measurements and subsequent conditional unitaries can make state preparation much faster~\cite{briegel2001persistent,raussendorf2005long,aguado2008creation,Piroli2020,tantivasadakarn2021long,bravyi2022adaptive,Lu2022}.
Here we elaborate how such measurements could be used in our algorithm.

\emph{Tree-RG circuit with measurements.---}Local measurements and conditional local unitaries are the standard framework to perform quantum teleportation~\cite{bennett1993teleporting,gottesman1999demonstrating}, which can be used to reduce the depth of the tree-RG circuit.
Isometries appearing in \cref{eq:RG_circuit} act on a constant number of sites which, although spatially separated, can be teleported at neighboring registers with a constant overhead.
This can be achieved by creating nearest-neighbor entangled pairs, then performing simultaneous measurements, and correcting (without postselection) based on the measurement outcomes~\cite{nielsen2002quantum} (this process is also detailed in Ref.~\cite{Lu2022}).

Therefore every isometry in \cref{eq:RG_circuit} takes constant time using measurement. Crucially, however, the tree-RG circuit requires only $O(\log \log (N/\epsilon))$ layers (in contrast to Ref.~\cite{Lu2022}).
Since the fixed-point state can be prepared in constant time as before, this gives a preparation algorithm for short-range correlated MPS with depth $O(\log \log (N/\eps))$.

\emph{Long-range MPS using measurements.---}Another consequence of including measurements is that the creation of GHZ-like states $\ket{\chi_M} = \sum_{i=1}^b \alpha_i \ket{i}^{\otimes M}$ becomes possible in only constant depth~\cite{briegel2001persistent,Piroli2021,verresen2021efficiently}.
These states are closely related to the fixed points of TI-MPS\footnote{Here we consider TI-MPS with constant $D$. Note that this does not include all states that are TI and have area law entanglement (e.g., $W$-state~\cite{perez2007matrix,cirac2021matrix}).} which, up to an isometry, take the form~\cite{cirac2017matrix} 
\begin{equation}\label{eq:fp_general_ti}
  \ket{\Omega'}  = \sum_{j = 1}^b \alpha_j^{(N)} \bigotimes_{i=1}^{N/q} \ket{\omega_{j}}_{R_i L_{i+1}}.
\end{equation}
The normal case corresponds to $b=1$ for which $\ket{\Omega'} = \ket{\Omega}$ while, in general, $b$ is upper bounded by the number of blocks in the canonical form and $\alpha_j^{(N)}$ may depend on $N$~\cite{cirac2017matrix}.

Importantly, the different $\ket{\omega_j}$ are orthogonal~\cite{cirac2017matrix}, which suggests a preparation procedure for $\ket{\Omega'}$. First create $\ket{\chi_{N/q}}$, which can be done in constant depth with measurements (following, e.g., Ref~\cite{Piroli2021}). Subsequently, apply in parallel the isometries  $W: \ket{j} \mapsto \ket{\omega_j}_{R_i L_{i+1}}$ such that $\ket{\Omega'} = W^{\otimes N/q}  \ket{\chi_{N/q}}$, which also takes constant depth.

In~\cite{sm} we show how to explicitly obtain a state of the form \cref{{eq:fp_general_ti}} that approximates well the target $\ket{\phi_N}$ up to local isometries by blocking $q \propto \log (N/\epsilon)$ sites. As a result, following the same steps as in the tree-RG circuit with measurements we have a scheme that approximates all TI-MPS (short- or long-range correlated) with depth $T=O(\log \log (N/\eps))$ [cf.~\cref{dep_meas}]. If instead measurements are only used for the preparation of $\ket{\chi_{N/q}}$, the depth is $O(\log (N/\eps))$. 

Our construction generalizes to inhomogeneous long-range correlated MPS exactly as in the short-range case.

\emph{Connection to MERA.---}The circuit in the tree-RG scheme can be interpreted as a finite-range MERA with $O(\log \log N)$ layers, namely a shallow tensor tree acting on the fixed-point state.
Specifically, the isometries ${V^{(i)}}$ [cf.~\cref{eq:RG_circuit}] are identified with the isometries in finite-range MERA, and all disentanglers are the identity, save for the first layer, which is identified with the single layer of unitaries that prepare the fixed-point state.
Hence, within the approximation error $\eps$,
\begin{equation}
     \textrm{normal TI-MPS} \subset
     \begin{array}{c}
        \textrm{finite-range MERA} \\  \textrm{with $O(\log \log N)$ layers.} 
    \end{array}
\end{equation}

\emph{Discussion and outlook.---}Our results also imply that MPS in the same phase can be transformed into each other using a log-depth circuit, in contrast to the well-known quasilocal evolution corresponding to polylogarithmic depth circuit~\cite{osborne2006efficient,haah2021quantum,coser2019classification,Schuch2011}.
It would be interesting to explore whether our results could be exploited for applications other than state preparation.
Specifically, a number of protocols~\cite{Cramer2010a,haghshenas2022variational,Ran2020a,lin2021real, dilip2022data,rudolph2022synergy} implicitly or explicitly depend on the ability to prepare (or disentangle) MPS using a sequential circuit. It may be possible to replace the sequential circuit with ours to reduce the circuit depth in these protocols. 
Another direction would be to extend our lower-bound proof and the preparation algorithm to prepare certain higher dimensional tensor network states~\cite{cirac2021matrix}.

\textit{Acknowledgments.---}%
We thank Yujie Liu, Miguel Fr\'ias P\'erez, and Rahul Trivedi for insightful discussions.
DM acknowledges support from the Novo Nordisk Fonden under grants No.~NNF22OC0071934 and No.~NNF20OC0059939. GS is supported by the Alexander von Humboldt Foundation.
The research is part of the Munich Quantum Valley, which is supported by the Bavarian state government with funds from the Hightech Agenda Bayern Plus. We acknowledge funding from the German Federal Ministry of Education and Research (BMBF) through EQUAHUMO (Grant No.~13N16066) within the funding program quantum technologies---from basic research to market.
The numerical calculations were performed using the ITensor Library~\cite{itensor}.

\bibliography{merged_bibliography,jabref,library,GS_bibliography}

\appendix

\setcounter{equation}{0}
\setcounter{figure}{0}
\setcounter{table}{0}
\makeatletter
\renewcommand{\theequation}{S\arabic{equation}}
\renewcommand{\thefigure}{S\arabic{figure}}

\subsection*{Proof of \cref{lm:1} and extension to non-normal tensors}

Here we first show how to explicitly obtain the approximate state $| \widetilde \phi_N \rangle$ for the non-normal case. Then we prove \cref{lm:1}', which bounds the approximation error both for normal and non-normal TI-MPS and thus immediately implies \cref{lm:1}. 

We consider general TI-MPS defined by
\begin{align}
	\ket{\phi_N} = \frac{1}{c_N} \sum_{i_1, \ldots, i_N}\Tr\left( A^{i_1}\cdots A^{i_N} \right)\ket{i_1\cdots i_N}
	\label{eq:TI-MPS2}
\end{align}
where $c_N>0$ is the normalization constant.
After a gauge transformation, every tensor $A$ can be expressed in terms of a basis of normal tensors~\cite{cirac2017matrix}
\begin{align}
	A^i = \bigoplus_{j=1}^b {\rm diag} (\mu_{j,1} ,\dots , \mu_{j,m_j}) \otimes  A_j^i
	\label{eq:CF}
\end{align}
where the normal $A_j$ are in canonical form~II and produce orthogonal vectors in the thermodynamic limit~\cite{cirac2017matrix}.
Without loss of generality we assume $| \mu_{j,k} | \le 1$ with at least one of them having magnitude exactly one.
The normal case thus corresponds to $b=1$ and a single $|\mu_{1,1}| = 1$.
From \cref{eq:CF}, it follows that the general form of a TI-MPS is
\begin{align} \label{eq:app_phi_general}
    \ket{\phi_N} = \frac{1}{c_N} \sum_{j=1}^b \beta_j  \ket{v_j},
\end{align}
where
\begin{align}
    \label{eq:beta}
    \beta_j = \sum_{k=1}^{m_j} \mu_{j,k}^N
\end{align}
and $\ket{v_j}$ is the (unnormalized) MPS generated by the normal tensor $A_j$ (i.e., \cref{eq:TI-MPS2} without $c_N$).

Let us now define the approximate state $| \widetilde \phi_N \rangle$, which generalizes \cref{eq:phi_tilde}.
For that, as in the normal case, we block $q$ sites and perform a polar decomposition of the tensor $B^{i_1 \dots i_q} = A^{i_1} \dots A^{i_q}$. This results in $B = V P$, where $P: \mathbb C^{D^2} \to \mathbb C^{D^2}$ is positive-semidefinite and the isometry satisfies $V^\dagger V = \Pi$ for $\Pi$ the projector onto the image of $P$.
Since $V$ is an isometry, $P$ inherits the block structure of $A$
\begin{align}
         \begin{array}{c}
		\begin{tikzpicture}[scale=.4,thick,baseline={([yshift=-4ex]current bounding box.center)}]
			\PTensor{0,0}
    	\draw (-.5,1.5) node {\scriptsize $i_1$};
    	\draw (0.65,1.5) node {\scriptsize $i_2$};
        \end{tikzpicture}
        \end{array}
        = \bigoplus_{j=1}^b
        {\rm diag} (\mu_{j,1}^q ,\dots , \mu_{j,m_j}^q) \otimes  
        P_{j}^{i_1i_2}
\end{align}
where $i_1,i_2 = 1 , \dots, D$ and all $P_{j}$ are normal tensors. We can therefore express
\begin{align}
    \ket{\phi_N} = \frac{1}{c_N} \big( \bigotimes_{i=1}^{N/q} V_i \big) \sum_{j=1}^b \beta_j   \ket{v_{{\rm pos},j}}.
\end{align}
where $\ket{v_{{\rm pos},j}}$ is the unnormalized MPS generated by $P_{j}$.

The approximate state $|\widetilde \phi_N \rangle$ is defined by replacing each normal tensor $P_j$ with its fixed-point counterpart $P_{j,\infty}$ (analogous to \cref{eq:B_TM}).
Equivalently, we replace each $\ket{v_{{\rm pos},j}}$ by the corresponding fixed-point state $\ket{\Omega_j}$. That is,
\begin{align} \label{eq:app_tidle_phi_1}
    |\widetilde \phi_N \rangle &= \frac{1}{\widetilde c_N} \big( \bigotimes_{i=1}^{N/q} V_i \big) \sum_{j=1}^b \beta_j \ket{\Omega_j}
\end{align}
where $\ket{\Omega_j} = \bigotimes_{i=1}^{N/q} \ket{\omega_j}_{R_i L_{i+1}}$ are (normalized) nearest-neighbor entangled pairs over $\mathbb C^{D^2}$.
Importantly, since a basis of normal tensors was used for the decomposition of \cref{eq:CF}, they satisfy local orthogonality $\langle \omega_j | \omega _{j'} \rangle = \delta_{j j'}$~\cite{cirac2017matrix}. 
\Cref{eq:app_tidle_phi_1} also justifies the form of \cref{eq:fp_general_ti}, for which
\begin{align}
    \alpha_j^{(N)} = \frac{\beta_j}{\sqrt{\sum_l |\beta_l|^2}},
\end{align}
where we explicitly added a superscript $(N)$ to remember that $\beta_j$ may have a decaying contribution from $| \mu_{j,k} | < 1$ [\cref{eq:beta}].
This contribution vanishes in the limit of blocking $q \to \infty$, but here is taken into account, because neglecting it leads to an unwanted additional error contribution.
Because of this, $\sum_j\beta_j\ket{\Omega_j}$ in \cref{eq:app_tidle_phi_1} may strictly speaking not be a fixed point of the RG transformation in the non-normal case.

We now turn to the error estimate. For that, the key lemma comes from Ref.~\cite{Piroli2021}, where it was shown that for the normal tensor $A_j$
\begin{align}  \label{eq:app_error}
 |  1 - |\langle \Omega_j  \ket{v_{{\rm pos},j}}| | =  O \left( \frac{N}{q} \exp(-\gamma q/\xi_{jj} )  \right)	
\end{align}
for all $0< \gamma< 1/2$ (see Eq.~(S29) in the Supplemental Material of Ref.~\cite{Piroli2021}). Here
\begin{align}
\xi_{jj} = - 1/ \ln | \lambda_2^{(j)} |     
\end{align}
denotes the associated correlation length, i.e., $\lambda_2^{(j)}$ the subleading eigenvalue of the transfer matrix of the normal tensor $A_j$.

We are now ready to state and prove our result.

\begin{L1}[Approximation error]
    Consider a sequence of TI-MPS $| \phi_N \rangle$ generated by the tensor $A$. Then for all $0 < \gamma < 1/2$:
    \begin{enumerate}[(i)]
        \item If $A$ is normal with correlation length $\xi$,
    \begin{equation} \label{eq:fid_err_gen_normal}
        \epsilon = O\left( \frac{N}{q} e^{-\gamma q/\xi} \right).
    \end{equation}
    \item For a general non-normal $A$,
    and $q=o(N)$
        \begin{equation} \label{eq:fid_err_gen_non_normal}
        \epsilon = O\left( \frac{N}{q} e^{-\gamma q/\xi_\mathrm{diag}} \right)
    \end{equation}
    where $\xi_\mathrm{diag} = \max_{j} \xi_{jj}$.
    \end{enumerate}
\end{L1}
\begin{proof}
\textbf{(i)} Let us start with the case of a normal tensor. As detailed in the main text,
\begin{align}
    | \widetilde \phi_N \rangle =  \bigotimes_{i=1}^{N/q} V_i | \Omega \rangle .
\end{align}
Then, by the triangle inequality
\begin{align*}
    \epsilon &=   1 - | \langle  \widetilde \phi_N   | \phi_N \rangle  | \\
    &\le \big|  1 - c_N |\langle  \widetilde \phi_N   | \phi_N \rangle  | \,\big| + | c_N - 1   |   \, |\langle \widetilde \phi_N   | \phi_N \rangle | .
\end{align*}
The first term is exactly equal to the LHS of \cref{eq:app_error}. Using Cauchy-Schwarz and that $c_N = \sqrt{\tr (E_A^N)}$, the second is $O(\exp(-N/\xi))$.
Since $q =o(N)$, the first term dominates.

\textbf{(ii)} We now move on to the non-normal case.
Here, another set of length scales $\xi_{jj'}$ plays a role, which is defined as follows.
Consider the inner product of two MPS over $N$ sites, $\langle v_j | v_{j'} \rangle$,
where $\ket{v_j}, \ket{v_{j'}}$ are generated by normal tensors $A_j$, $A_{j'}$ that belong to different basis elements [cf. \cref{eq:CF,eq:app_phi_general}].
Then
\begin{align}\label{eq:app_decay_mixed}
    | \langle v_j | v_{j'} \rangle | = O \left( e^{-N / \xi_{jj'}} \right)
\end{align}
where 
\begin{align}
    \xi_{jj'} = - 1 / \ln \tau_{\max}
\end{align}
where $\tau_{\max}$ the spectral radius of the mixed transfer matrix $E_{jj'}=\sum_iA_j^{i*}\otimes A_{j'}^i$ and $\tau_{\rm max}<1$ (see Lemma~A.2 in \cite{cirac2017matrix}).

Using triangle and Cauchy-Schwarz inequalities,
\begin{align*}
    \epsilon =  &  1 - | \langle  \widetilde \phi_N   | \phi_N \rangle | \\
    \le  & \big| 1 -  \frac{c_N}{\widetilde c_N}  | \langle  \widetilde \phi_N   | \phi_N \rangle | \, \big| + \big| \frac{c_N}{\widetilde c_N} - 1 \big||\langle \widetilde \phi_N   | \phi_N \rangle |.
\end{align*}
For the first term, we have
\begin{align*}
    &\big | 1 -  \frac{c_N \widetilde c_N}{\widetilde c_N^2} | \langle  \widetilde \phi_N   | \phi_N \rangle | \, \big| \le \\
     &  \frac{\sum_j |\beta_j|^2 | 1 - \langle \Omega_j | v_{{\rm pos},j} \rangle |}{\sum_{l} |\beta_l|^2} 
     + \left|  \frac{\sum_{jj'} \beta_j^* \beta_{j'}  \langle \Omega_j | v_{{\rm pos},j'} \rangle}{\sum_{l} |\beta_l|^2} \right|,
\end{align*}
where we used that $\widetilde c_N^2=\sum_j|\beta_j|^2$.
To bound the first fraction, we use \cref{eq:app_error} and get $O \left( \frac{N}{q} \exp(-\gamma q/\xi_\mathrm{diag} )  \right)$ where $\xi_\mathrm{diag} = \max_{j} \xi_{jj}$.
By \cref{eq:app_decay_mixed} the second fraction is $O(\exp(-N/ \xi_\mathrm{off-diag}))$ where $\xi_\mathrm{off-diag} = \max_{j \ne j'} \xi_{jj'}$.
The remaining term is
\begin{align*}
    &\big| c_N/\widetilde c_N - 1 \big| \le | c_N^2 / \widetilde c_N^2 -1 | \\ 
    \le & \left| \frac{\sum_j |\beta_j|^2 (1 - \langle v_j | v_j \rangle)}{\sum_l|\beta_l|^2} \right|
    + \left|\frac{\sum_{j \ne j'} \beta_j^* \beta_{j'}  \langle v_{j'} | v_{j} \rangle}{\sum_l|\beta_l|^2} \right| .
\end{align*}
The terms of the first sum are $O(\exp(-N/ \xi_{jj}))$, while those of the second sum $O(\exp(-N/ \xi_{jj'}))$.

Putting everything together, we get
\begin{equation}
    \epsilon = O\left( \frac{N}{q} e^{-\gamma q/\xi_\mathrm{diag}} \right)
    + O \left(  e^{-N/\xi_\mathrm{off-diag}}\right),
\end{equation}
where $\xi_\mathrm{diag} = \max_{j} \xi_{jj}$ and $\xi_\mathrm{off-diag} = \max_{j \ne j'} \xi_{jj'}$.
If we further assume $q=o(N)$, the second contribution disappears, giving \cref{eq:fid_err_gen_non_normal}.
\end{proof}

\subsection*{Proof of \cref{th:1}}

Before we present the proof, let us introduce the following lemma, which we will use to distinguish the states based on the mismatch between states with strictly finite correlation length and states with exponentially decaying correlations. 
\begin{lemma}[Exponentially decaying correlations]
    \label{lm:2}
    Let $\{\ket{\phi_N}\}$ be a sequence of TI-MPS generated by an injective tensor $A$ with finite correlation length $\xi>0$ [cf.~\cref{eq:TI-MPS2}].
    Then, we can always find two local operators $\O_1,\O'_s$ acting on spins $1$ and $s$ with $||\O||=||\O'||=1$ such that for any integer $s>1$ and sufficiently large $N$,
    \begin{subequations}   
        \begin{eqnarray}
            \label{auxform1}
            \langle \phi_{N}|\O_1|\phi_{N}\rangle =  \langle \phi_{N}|\O_s'|\phi_{N}\rangle &=& 0,\\
            \label{auxform2}
            \langle \phi_{N}|\O_1\O_s'|\phi_{N}\rangle  &\ge& c e^{-(s-1)/\xi}
        \end{eqnarray}
    \end{subequations}
    where $c>0$ is independent of $N,s$.
\end{lemma}

\begin{proof}
Consider the connected correlation function
\begin{equation}
    \Delta= \langle \phi_N|\O_1\O_s'|\phi_N\rangle - \langle \phi_N|\O_1|\phi_N\rangle \langle \phi_N|\O_s'|\phi_N\rangle,
\end{equation}
where $\O_1$ and $\O_s$ are two (potentially different) operators placed at sites $1$ and $s$.
We have 
\begin{equation}
    \begin{aligned}
        \Delta &= \frac{1}{c_N^2} [\Tr(E_\1^{N-s-1} E_\O E_\1^{s-1} E_{\O'})\\
        &- \Tr(E_\1^{N-1} E_\O)\Tr(E_\1^{N-1} E_{\O'})],
    \end{aligned}
\end{equation}
where the normalization $ c_N = \sqrt{\Tr(E_\1^{N})}$, and 
\begin{equation}
    E_\Q = \sum_{i,j=1}^d \langle i|\Q|j\rangle \; (A^i)^*\otimes A^j,\quad \Q\in\{\1,\O,\O'\}
\end{equation}
where $d$ is the physical dimension of the MPS.
Given the spectrum of $E_\1$ we can always take $N$ sufficiently large so that we can approximate with an arbitrarily small error,
$E_\1^{N-1} = E_\1^{N-s-1}=|R_1\rangle\langle L_1|+O(e^{-N/\xi})$,
where $\xi=-1/\ln(|\lambda_2|)$ and $1 = \lambda_1 > | \lambda_2| > \dots$ are the eigenvalues of $E_\1$.
Note that in the main text we use the gauge in which $\ket{R_1}=\ket\rho$ and $\ket{L_1}=\ket\1$.
Here,
$\langle L_1|R_1\rangle=1$, so that $c_N\approx 1$ and 
\begin{equation}
    \Delta \approx \sum_{i=2}^{D^2} \lambda_i^{s-1} \langle L_1|E_\O|R_i\rangle \; \langle L_i|E_{\O'}|R_1\rangle,
\end{equation}
where $D$ is the bond dimension and 
we have written
$E_\1^{s-1}=\sum_i \lambda_i^{s-1}\ket{R_i}\bra{L_i}$.
Since the tensor $A$ is injective~\cite{perez2007matrix}, we can always choose $\O$ (and $\O'$), such that the corresponding transfer matrix $E_\O=|A\rangle\langle B|$ for arbitrary $A,B$ (up to a normalization constant).
In particular, we can use this to impose that
\begin{subequations}
\begin{eqnarray}
	\langle L_i|E_\O |R_i\rangle= \langle L_i|E_{\O'} |R_i\rangle &=& 0, \quad \forall i\\
	\langle L_1|E_\O |R_i\rangle= \langle L_i|E_{\O'} |R_1\rangle &=& 0, \quad \forall i> 2,\\
	\langle L_1|E_\O |R_2\rangle\langle L_2|E_{\O'} |R_1\rangle&=&c'>0.
\end{eqnarray}
\end{subequations}
The first line ensures (\ref{auxform1}), while the second and third ensure (\ref{auxform2}) for sufficiently large $N$, with $c=c'/2$, where $1/2$ is an arbitrary constant chosen for concreteness.
\end{proof}

Now, we can prove \cref{th:1}.
Let $\{\ket{\phi_N}\}$ be a sequence of TI normalized MPS on $N$ sites generated by a normal tensor $A$, and $\{\ket{\psi_N}\}$ a sequence of states obtained by applying a depth-$T$ local quantum circuit to a product state and define the error $\eps=1-|\bra{\phi_N}\psi_N\rangle|$.
\begin{T1}[restated]
    If $T=o(\log N)$ there is some $N_0$ such that for all $N>N_0$ we have $\eps >1/2$.
\end{T1}
\begin{proof}
Let us assume that $T=o[\log(N)]$ and $T>2\xi$,
since we can always add layers of identity operators to increase the depth of the circuit.
We approximate $\{\ket{\phi_N}\}$ through $\{\ket{\widetilde\phi_N}\}$ [cf.~\cref{eq:phi_tilde}] obtained by blocking $q_N=\lceil 2(1+\eta)\xi\ln N \rceil$ with $\eta>0$ and use \cref{lm:1} to bound the error as
\begin{equation}
\label{eq:tilde_bound}
    \eps=1-|\bra{\widetilde\phi_N}\phi_N\rangle|
    <c_0N^{-\eta}
\end{equation}
for some constant $c_0$ independent of system size.
We take $N$ such that we have a large number of blocks, all of the same size, $q_N$, except for the last one, which may be larger.
This is always possible, as $q_N=O(\log N)$.
We also take $N$ large enough to ensure $q_N>T$.

We thus have
\begin{equation}
    \label{auxdp}
    d(\phi_N,\psi_N) \ge d(\psi_N,\widetilde\phi_N)-\sqrt{2c_0}N^{-\eta/2},
\end{equation}
where $d(\rho,\sigma)=||\rho-\sigma||_1/2$ is the trace distance~\cite{nielsen2002quantum},
as well as an upper bound on trace distance from fidelity combined with \cref{eq:tilde_bound}
\begin{equation}
    \label{eq:tilde_distance}
    d(\phi_N,\widetilde\phi_N)\le \sqrt{1- |\langle \phi_N |\widetilde\phi_N \rangle|^2}
    \le \sqrt{2c_0}N^{-\eta/2}.
\end{equation}
In the following, we will find a lower bound to the first term in \cref{auxdp} to make the difference larger than 1/2.
We will also drop the subscript $N$ to simplify notation.

To obtain a bound on the distance of $\ket{\psi_N}$ and $\ket{\widetilde\phi_N}$, let us consider instead a suitable subsystem.
To that end, we divide the chain into $\lfloor N/(2q)\rfloor$
blocks of size $2q$ each, with the last block potentially smaller than $2q$.
We then trace over all $2q$ spins at the sites contained in the intervals $[4mq+1,2(2m+1)q]$, with $m=0,1,\ldots$ in both states $\ket{\widetilde\phi}$ and $\ket{\psi}$.
In case the last block we constructed is smaller than $2q$, we trace it as well.
If we perform such an operation on $|\widetilde\phi\rangle\langle \widetilde\phi|$, we obtain a product state 
\begin{equation}
    \rho=\rho_0^{\otimes k},
\end{equation}
which follows from the definition of $\ket{\widetilde\phi_N}$ and the fact that it is invariant under translation by $q$ sites.
We have 
\begin{equation}\label{eq:k}
    k = \lfloor N/4q\rfloor.
\end{equation}

Analogously, applying the same trace to $|\psi\rangle\langle \psi|$
we also obtain a product state, because $q>T$,
\begin{equation}
    \sigma=\sigma_1\otimes\ldots\otimes \sigma_k.
\end{equation}
Using the fact that the trace distance is contractive under tracing, and bounding it in terms of the Uhlmann fidelity, we have~\cite{nielsen2002quantum}
\begin{equation}
	\label{auxdp3}
	d(\widetilde\phi,\psi) \ge d(\rho,\sigma) \ge 1- F(\rho,\sigma),
\end{equation}
with the Uhlmann fidelity between two density matrices $\rho$ and $\sigma$ defined as $F(\rho, \sigma)=\operatorname{Tr} \sqrt{\sqrt{\rho} \sigma \sqrt{\rho}}$.

Given that $\rho$ and $\sigma$ are product states, we have
\begin{equation}
    \label{inequax}
    F(\rho,\sigma)=\prod_{i=1}^k F(\rho_0,\sigma_i) \le (1-\delta)^{k/2}
\end{equation}
where $\delta={\rm min}_i d(\rho_0,\sigma_i)^2$, and where we have used another bound between the fidelity and the trace distance~\cite{nielsen2002quantum}.

Next we will lower bound $\delta$ using \cref{lm:2}.
However, we have to be a bit careful since this lemma applies to $\phi$ instead of $\widetilde\phi$.
Fortunately, we can use \cref{eq:tilde_distance} to replace one with the other. For the sake of concreteness, we will bound $d(\rho_0,\sigma_1)$ but the same analysis applies to every $\sigma_i$.

Let us take $s=2T+1$, and the operators $\O,\O'$ from \cref{lm:2} to define
\begin{eqnarray}
	a_\Q&=&\langle \phi|\Q|\phi\rangle, \\
	\widetilde a_\Q&=&\Tr(\rho_0 \Q)=\langle \widetilde\phi|\Q|\widetilde\phi\rangle, \\
	b_\Q&=&\Tr(\sigma_1 \Q)=\langle \psi|\Q|\psi\rangle,
\end{eqnarray}
where $\Q\in\{\O_1, \O_s',\O_1\O_s' \}$.
Given that $||\O||=||\O'||=1$, we can bound
\begin{equation}
    \label{auxbound}
    d(\rho_0,\sigma_1) \ge \max_{\Q=\O_1,\O'_s,\O_1\O_s'}\big||\widetilde a_\Q|-|b_\Q|\big|.
\end{equation}

According to \cref{lm:2}, $a_{\O_1}=a_{\O'_s}=0$, and 
\begin{equation}\label{eq:mu}
	a_{\O_1\O'_s}> \mu= c_1 e^{-2T/\xi}.
\end{equation}
In order to use \cref{auxbound}, we need to connect $\widetilde a_\Q$ to $a_\Q$.
Using \cref{eq:tilde_distance}, we choose sufficiently large $N$ to obtain $d(\widetilde\phi,\phi) \le \mu/3$.
This immediately implies that $|\widetilde a_\O|,|\widetilde a_{\O'}| < \mu/3$ and $\widetilde a_{\O_1\O'_s} > 2 \mu/3$.
Moreover, since $\psi$ is created from a product state by a depth-$T$ circuit, every connected correlation for operators at a distance larger than $2T$ vanishes.
Since $s=2T+1$ we therefore have $b_{\O_1\O'_s}=b_{\O_1} b_{\O_s'}$.
We now show that for any choice of $b_{\O_1}$, $ b_{\O_s'}$, the distance $d(\rho_0, \sigma_1)$ is bounded below by a constant.
Thus, we minimize \cref{auxbound} with respect to $x_1=b_{\O_1}$ and $x_2=b_{\O_s'}$, imposing $\mu<1/2$:
\begin{equation}\label{dpaux2}
\begin{aligned}
    d(\rho_0,\sigma_1) &\ge \min_{x_1,x_2}\left[\max(|x_1-\frac{\mu}{3}|,|x_2-\frac\mu3|,|\frac{2\mu}{3}-x_1x_2|)\right] \\
    &= \min_{x}\left[\max(|x-\mu/3|,|2\mu/3-x^2|)\right]>\mu/3.
\end{aligned}
\end{equation}
From this it immediately follows that $\delta> \mu^2/9$. 

Putting \cref{dpaux2} and \cref{inequax} into \cref{auxdp}, we arrive at
\begin{equation}
	d(\phi,\psi) > 1-(1-\mu^2/9)^{k/2} - \sqrt{2c_0}N^{-\eta} .
\end{equation}
The last term is negligible for large $N$.
If then $k>2/\delta$, one has that $(1-\delta)^{k/2}< 1/e$ and thus $d(\phi,\psi)>\sqrt{3/4}$, which implies $\eps>1/2$.

So we need to find $N$ for which $k>2/\delta$.
Using \cref{eq:k} and the definition of $q$, we have $k>N/(10\xi\ln N)$, where we chose $\eta<1/4$.
Using $\delta>\mu^2/9$ and the definition of $\mu$ (\cref{eq:mu}), we find $(18e^{4T/\xi})/c_1^2>2/\delta$.
Putting this together, we need to find $N$ such that 
\begin{equation}
    k>\frac{N}{5q}>\frac{N\gamma}{10\xi\log N}>\frac{18}{c_1^2}e^{4T/\xi}>\frac{2}{\delta}.
\end{equation}
Since $T=o[\log(N)]$,
we can always find an $N_0$ such that this is fulfilled for all $N>N_0$.
\end{proof}

\subsection*{Numerical results}

Here we provide numerical evidence that our proposed strategies are capable of preparing relevant short-range correlated (inhomogeneous) MPS under open boundary conditions, which are of the form
\begin{align}
	\ket{\phi_N}_{\rm obc} = \sum_{i_1,i_2, \ldots, i_N} A_{[1]}^{i_1} A_{[2]}^{i_2} \cdots A_{[N]}^{i_N} \ket{i_1 i_2\cdots i_N},
	\label{eq:obc-MPS}
\end{align}
where $A_{[1]}^{i_1}$ and $A_{[N]}^{i_N}$ are tensors on the boundary, while the tensors in the bulk are defined similarly to those in \cref{eq:TI-MPS}. To prepare such states, we construct the isometries analytically, as delineated in the main text, and use a simple extension of the Evenbly-Vidal algorithm~\cite{evenbly2009algorithms,lin2021real} to variationally find the fixed-point state of the form in \cref{eq:fp_general_ti} that maximize the fidelity between the approximate state [cf.~\cref{eq:phi_tilde}] and the target state (see more details in \cref{sec:vari_optm}). Since the variational space comprises only (superpositions of) product states of entangled pairs [cf.~\cref{eq:fp_general_ti}], this optimization scheme is found to be highly efficient. Moreover, this local optimization strategy can effectively encapsulate the inhomogeneity present in the target state, which makes it especially suitable for preparing non-TI MPS and states with open boundary conditions, and can be directly extended to the case of long-range MPS.

In the following, we numerically study the performance of our algorithm for preparing three types of short-range correlated MPS: (1) the 1D AKLT state~\cite{Affleck1987, Affleck1988}, (2) an MPS family with tunable correlation length~\cite{Wolf2006}, 
and (3) inhomogeneous random MPS~\cite{garnerone2010typicality}, which illustrate that our algorithm succeeds in general inhomogeneous settings, too.

\subsection{Preparation of AKLT state and the MPS family~\cite{Wolf2006}}
The 1D AKLT state is a paradigmatic state in condensed matter physics, with important application in measurement-based quantum computation~\cite{Brennen2008}. The 1D AKLT state can be formed by first having a product state of singlets consisting of virtual qubits that connect neighboring sites of the 1D chain, then projecting the virtual qubits of two neighboring pairs to their symmetric subspace (with spin $S=1$). In our calculation, we consider the spin $S=1$ at each site is formed by the symmetric subspace of two actual qubits in the quantum device, and the resulting AKLT state is an MPS of bond dimension $D=2$ and physical dimension $d=4$.

Figure \ref{fig_numerics}(a) illustrates the scaling of the error per block $\eps/M$ (where the number of blocks $M = N/q$)[cf.~\cref{eq:fid_err}] as a function of blocking range $q$ (measured in units of the correlation length $\xi$, with $\xi_{\rm AKLT} = 1/ \ln 3$) for the AKLT state. In our calculations, we chose the number of blocks $M=200$. As expected, both our circuit constructions have comparable performance, with $\eps/M$ exhibiting an exponential decay with $q/\xi$, in accordance with the bound of \cref{eq:fid_err}. This verifies the predicted scaling of $T=O(\log (N/\eps))$.

To further assess the efficacy of our algorithm for MPS with varying correlation lengths $\xi$, we also investigate the MPS class of bond (physical) dimension $D=2$ ($d=2$), with matrices in \cref{eq:obc-MPS} of the form~\cite{Wolf2006}
\begin{equation} \label{g_mps}
	A_{[j]}^{0} =\left(\begin{array}{ll}
0 & 0 \\
1 & 1
\end{array}\right), \quad A_{[j]}^{1} =\left(\begin{array}{ll}
1 & g \\
0 & 0
\end{array}\right), \forall j \in (2,...,N-1),
\end{equation}
and the boundary tensors are chosen as the $2\times 2$ identity matrix. The correlation length of this MPS class can be tuned by the parameter $g$ as $\xi=\left|\left(\ln \frac{1-g}{1+g}\right) \right|^{-1}$. The results on the scaling of the error per block $\eps/M$ is shown in \cref{fig_numerics}(a). We see that, for the AKLT state and the MPS class of various correlation lengths $\xi \approx 4, 16$, the scaling of $\eps/M$ show almost the same behavior as $\eps/M \sim \exp(-\gamma_{\rm num} q/\xi)$ with  $\gamma_{\rm num} \approx 2$, where the number of blocks $M=N/q$. Notably,  $\gamma_{\rm num}$ is much larger than the analytically derived value $0<\gamma<1/2$ [cf.~\cref{eq:fid_err}]. Therefore, for these two classes of states, in practice, one can prepare them faster than predicted in the worst-case bound \cref{eq:fid_err}.

\subsection{Scaling of CNOT depth for various schemes}
\label{cnot_sec}
To make the scaling $T=O(\log N/\eps)$ of our protocol more relevant to the current devices where the multi-qubit unitaries are decomposed into CNOT gates and single-qubit rotations, we present a simple comparative study for the scaling of the CNOT depth $T_{\rm CNOT}$ required to prepare the MPS class [cf.~ \cref{g_mps}] with correlation length $\xi \approx 4$ with the fidelity ${\cal F}=|\langle \phi_N |\widetilde\phi_N \rangle|^2=0.9$.
We compare four different schemes: (1) sequential-RG scheme, (2) tree-RG scheme, (3) tree-RG scheme assisted by measurements, and (4) the sequential scheme~\cite{Schoen2005}. We present the result directly here and provide an explanation of our estimation of $T_{\rm CNOT}$ later.

Figure \ref{fig_numerics}(b) shows $T_{\rm CNOT}$ for these four different schemes. As anticipated, for both the sequential-RG and the tree-RG schemes, we observe the overall scaling $T=O(\log N)$, and these methods result in a significantly smaller $T_{\rm CNOT}$ compared to that of the sequential method, particularly for large system sizes $N$. Additionally, due to the fact that the blocking range $q$ can only increase discretely, we observe plateau-like features in the scaling of $T_{\rm CNOT}$ in \cref{fig_numerics}(b). Moreover, when the tree-RG scheme further assisted by measurements, the $T_{\rm CNOT}$ is further reduced compared to the stand-alone tree-RG scheme, yielding the smallest $_{\rm CNOT}$ among all schemes, with only $T_{\rm CNOT}\approx 100$ when creating this state of $N=10^6$ qubits.

\begin{figure}[h!]
\centering
\includegraphics[width=0.48\textwidth]{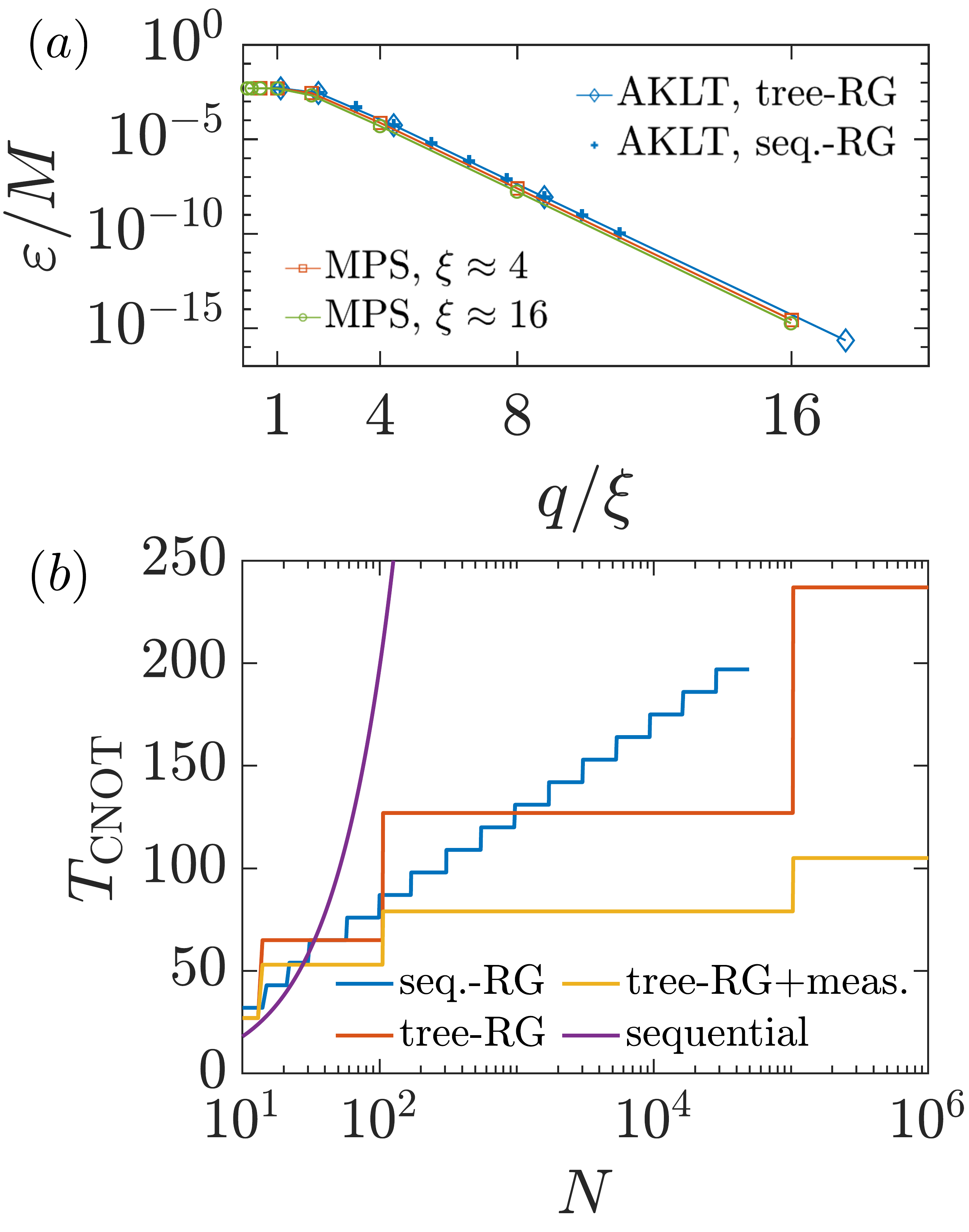}
\caption{(a) The error per block $\eps/M$ [cf.~\cref{eq:fid_err}] plotted as a function of the blocking range $q/\xi$, for preparing states in the MPS family~\cite{Wolf2006} and the AKLT state. (b) Estimated CNOT depth $T_{\rm CNOT}$ for preparing the MPS family of correlation length $\xi \approx 4$ with fidelity ${\cal F} = 0.9$ using different schemes. The `sequential-RG' denote the sequential-RG circuit, `tree' denote the tree-RG circuit. The `tree + meas.' denote the measurement-assisted tree-RG circuit, and 'sequential' refers to the sequential circuit~\cite{Schoen2005}.}
        \label{fig_numerics}
\end{figure}

\subsubsection{Estimation of the CNOT depth}
\label{esti_cnot}

This technical subsection elaborates on how we estimate $T_{\rm CNOT}$ for different schemes studied in \cref{fig_numerics}(b). For the sake of simplicity, our focus lies on the MPS with the bond (physical) dimension $D=2$ ($d=2$), which aligns with the results presented in \cref{fig_numerics}(b). Estimation for MPS with varying physical or bond dimensions can be accomplished in a similar manner.

In general, the schemes considered in \cref{fig_numerics}(b) involve two types of gates: (1) Isometries mapping from $m$ qubits to $n\ge m$ qubits and (2) SWAP gates utilized to change the qubit location of the fixed-point state in the sequential-RG scheme and to implement long-range isometries in the tree-RG scheme.

It is well known that a SWAP gate can be decomposed as three CNOT gate. For the isometries (1), we estimate the CNOT depth of each isometry, denoted by $T_{\rm iso}(m,n)$, using its theoretical lower bound~\cite{PhysRevA.93.032318}
\begin{equation} \label{}
	T_{\text {iso }}(m, n) = \lceil \frac{1}{4}\left(2^{n+m+1}-2^{2 m}-2 n-m-1\right) \rceil .
\end{equation}
This theoretical limit could potentially be reached by existing gate decomposition approaches~\cite{rakyta2022approaching}. Specifically, for preparing MPS with $d=D=2$, we have $T_{\rm iso}(1,2) = 2$ for the sequential scheme, $\quad T_{\rm iso}(2,3) = 10$ for the sequential-RG scheme, and $\quad T_{\rm iso}(2,4) = 26$ for the tree-RG scheme.

In the following we count the gates used in each scheme for preparing MPS of system size $N$ and bond (physical) dimension $D=2 (d=2)$:
\begin{itemize}
    \item The sequential scheme uses one layer of 0-to-2 qubit isometry on the boundary and $N-2$ layers of 1-to-2 qubit isometries in the bulk~\cite{Schoen2005,Schon2007}.
    \item The sequential-RG scheme with blocking range $q$ uses $q-2$ layers of SWAP gates and $q-2$ layers of 2-to-3 qubit isometries. 
    \item The tree-RG scheme iteratively blocks the chain for $\sim \log(q)$ times. The first blocking (blocking to injectivity) produce a layer of two-qubit gates. After that, the $m$-th blocking ($m \ge 2$) produce a layer of 2-to-4 qubit isometries, where the largest distance between qubits within the same isometry is $2^m$. We implement such long-distance isometries using local gates by first swapping the qubits to the center region of the isometry, then local implementing the 2-to-4 isometry, and finally swapping the qubits back. This leads to an additional $(2^{m} - 4)$ layers of SWAP gates for each $m$.
    \item The measurement-assisted tree-RG scheme simply eliminates the SWAP cost in the aforementioned tree-RG scheme, since now the long-distance isometries can be executed by gate teleportation~\cite{bennett1993teleporting,gottesman1999demonstrating}. It's noteworthy that there are also Bell-state ancilla preparation, measurement, and postprocessing costs involved in this scheme, but our focus here is solely on the circuit depth of the scheme.
\end{itemize}

Based on the above components, we can directly estimate $T_{\rm CNOT}$ as a function of the system size $N$ and the required blocking range $q$. Here, $q$ can be obtained from the scaling of the error per block (analogous to that in \cref{fig_numerics}(a)) and the required state preparation fidelity ${\cal F}$, thereby resulting in the $T_{\rm CNOT}$ shown in \cref{fig_numerics}(b).

\subsection{Preparation of inhomogeneous random MPS}

Here we illustrate our algorithm in the fully inhomogeneous case of random MPS.
Since any MPS can be brought to a canonical form with isometric tensors, a natural way to define the corresponding ensemble is by choosing each tensor randomly according to the Haar measure of the unitary group $U(dD)$~\cite{garnerone2010typicality}.
As these states correspond to the ground states of disordered local Hamiltonians~\cite{cirac2021matrix}, 
they can be considered as representative states of the trivial topological phase~\cite{chen2011classification,Schuch2011}, and it is of interest to explore various properties of this class~\cite{garnerone2010typicality,haferkamp2021emergent,lancien2022correlation,Haag2023typical}.

Since random MPS are expected to be short-range correlated~\cite{Haag2023typical}, we anticipate that our protocol can efficiently prepare this class of states. In \cref{fig_disorder}, we display the scaling of the error per block $\eps/M$ with the blocking range $q$ for randomly sampled 1000 states of $d=D=2$ using the tree-RG protocol, and observe the asymptotic scaling
\begin{equation} \label{disorder_err}
\eps/M \sim \exp(-c q)
\end{equation}
in all instances (note that the number of blocks $M=N/q$), with $c$ varying only slightly between different individual cases. This scaling is reminiscent of the behavior predicted in \cref{eq:fid_err}, and it directly implies that our protocols can prepare such inhomogeneous random MPS efficiently.

\begin{figure}[h!]
	\centering
	\includegraphics[width=0.48\textwidth]{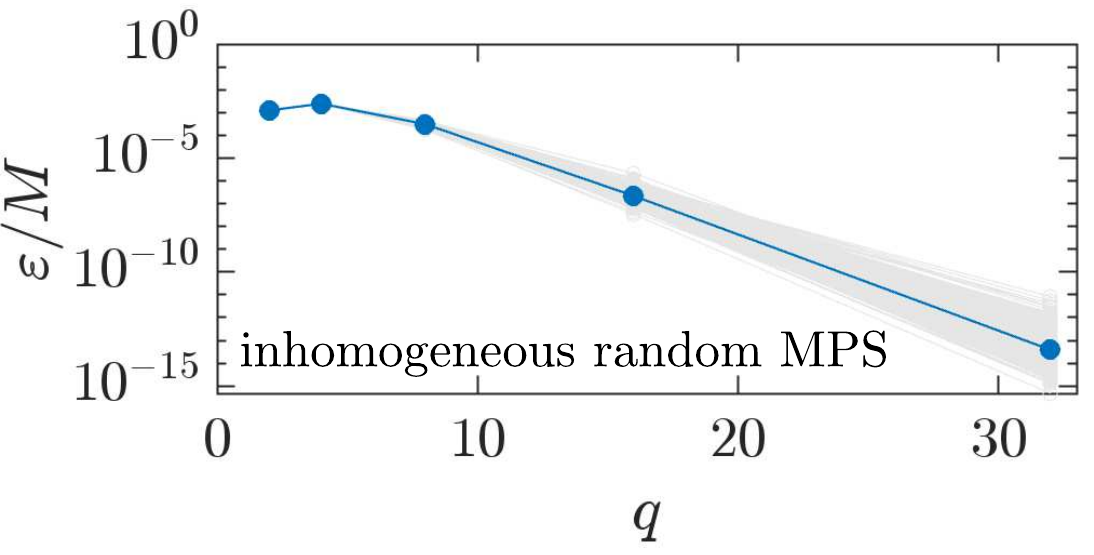}
        \caption{Error per block $\eps/M$ as a function of the blocking range $q$ for inhomogeneous random MPS of bond dimension $D=2$. The gray lines represent a total of 1000 individual samples, while the blue line illustrates the average behavior (averaged over $c$ in \cref{disorder_err}).}
        \label{fig_disorder}
\end{figure}


\subsection{The local variational optimization}
\label{sec:vari_optm}

Here we provide more details of the optimization algorithm developed in Ref.~\cite{evenbly2009algorithms,lin2021real}, which is used here for preparing inhomogeneous MPS.

Our task is to prepare the MPS $\ket{\phi_N}_{\rm obc}$ [cf.~\cref{eq:obc-MPS}] of a finite physical dimension $d$ and bond dimension $D$. As described in the main text, we first block each neighboring $q$ site to analytically construct the isometries. This leaves us with $\left|\phi_{\mathrm{pos}}\right\rangle$ [cf.~\cref{eq:phi_pos}], which we aim to approximate with a state $|\Omega\rangle$ consisting of nearest-neighbor entangled pairs
\begin{equation} \label{eq:inhomo_fp}
|\Omega\rangle=\bigotimes_{i=1}^{N / q}\left|\omega^i\right\rangle_{R_i L_{i+1}}.
\end{equation}

The optimization algorithm aims to maximize the fidelity $\mathcal{F}=|\langle\phi_{\mathrm{pos}}\mid\Omega\rangle|^2$ by optimizing the form of nearest-neighbor entangled pairs $\{ |\omega^i\rangle_{R_i L_{i+1}} \}$ in $|\Omega\rangle$. For this, we write each $|\omega^i\rangle_{R_i L_{i+1}}$ as a local two-qudit unitary $W_i$ acting on the product state, as 
\begin{equation} \label{}
	\left|\omega^i\right\rangle_{R_i L_{i+1}} = W_i \left|00\right\rangle_{R_i L_{i+1}}.
\end{equation}
To optimize the $i$-th unitary $W_{i}$, we write the overlap between $\left|\phi_{\mathrm{pos}}\right\rangle$ and $| \Omega \rangle$ as
\begin{equation} \label{}
\langle\phi_{\mathrm{pos}}\mid\Omega\rangle=\langle\phi_{{\rm pos}}|W_{i}\underset{|\Omega_{i}\rangle}{\underbrace{\prod_{j\neq {i}}^{N/q}W_{j}\bigotimes_{j=1}^{N/q}\left|00\right\rangle _{R_{j}L_{j+1}}}}={\rm Tr}[E_{i}W_{i}],
\end{equation}
where $E_{i} = |\Omega_{i}\rangle \langle\phi_{\mathrm{pos}}|$ is the \textit{environment} of the unitary $W_{i}$. By absorbing the overall phase factor into the unitary $W_{i}$, the fidelity can be expressed as the square of the real part of the overlap, as
\begin{equation} \label{}
	{\cal F}_{\rm pos}= {\rm Re}[\langle \phi_{\mathrm{pos}}\mid\Omega\rangle]^{2} ={\rm Re}[{\rm Tr}[E_{i}W_{i}]]^{2}.
\end{equation}
Therefore, one can directly find the unitary $W_{i}$ to maximize ${\cal F}_{\rm pos}$ by doing a singular value decomposition of the environment $E_{i}=X_{i}S_{i}Y_{i}^{\dag}$, and choosing $W_{i} = Y_{i}X_{i}^\dag$. By iterating through all the gates $\{W_{i}\}$ and sweeping back and forth until the fidelity $\cal F$ converges, we find the $| \Omega \rangle$ that best approximate $\left|\phi_{\mathrm{pos}}\right\rangle$~\cite{evenbly2009algorithms,lin2021real}.

Since there are $N/q$ commuting two-qudit unitaries $\{W_i\}$ used for creating $|\Omega \rangle$ and the target $|\phi_{\mathrm{pos}} \rangle$ is an MPS of finite physical dimension and bond dimension $D$, the construction of environment $\{ E_i \}$ and its SVD decomposition are computationally efficient, with the computational effort of constructing each $E_i$ scales as $O(ND^4)$. Moreover, we do not encounter the problem of falling into local minima for obtaining the results shown in \cref{fig_numerics,fig_disorder}, thanks to the small number of variational parameters in $|\Omega \rangle$. Overall, this allows the variational optimization to efficiently find the best approximation $|\Omega \rangle$ to $|\phi_{\mathrm{pos}} \rangle$. We also point out that there are more advanced methods for optimizing isometries~\cite{hauru2021riemannian}, which might be useful for the optimization task here as well. Finally, for the case of long-range correlated MPS, the fixed point $|\Omega' \rangle$ [cf.~\cref{eq:fp_general_ti}] can be created by applying a single layer of local two-qudit unitaries on the underlying GHZ-type state, and the optimization algorithm described here can be directly extended to that case to optimize those two-qudit unitaries.

\end{document}